\documentclass[]{amsart}
\usepackage[margin=1in]{geometry}
\usepackage[dvipsnames]{xcolor}
\usepackage{amsthm}
\usepackage{hyperref}
\usepackage{enumitem}
\usepackage[backend=bibtex8, maxnames=50]{biblatex}
\usepackage{amsmath}
\usepackage{quiver}
\usepackage{comment}
\usepackage{amssymb}
\usepackage{subfig}
\usepackage{tikz}
\usepackage{bbold}
\usepackage{graphicx}
\hypersetup{
    colorlinks=true,
    citecolor=blue,
    linkcolor=blue,
    filecolor=magenta,      
    urlcolor=cyan,
    }

\newcommand{\delz}[0]{\partial_z}

\newcommand{\delbar}[0]{\Bar{\partial}}
\newcommand{\delbarz}[0]{\partial_{\bar{z}}}
\newcommand{\zbar}[0]{{\Bar{z}}}

\newcommand{\de}[0]{\textrm{d}}
\newcommand{\Jac}[0]{\textrm{Jac}}
\newcommand{\Tot}[0]{\textrm{Tot}}
\newcommand{\sub}[0]{\subset}
\newcommand{\Tr}[0]{\textrm{Tr}}

\newcommand{\End}[0]{\mathrm{End}}
\newcommand{\irr}[0]{\mathrm{irr}}
\newcommand{\Hom}[0]{\mathrm{Hom}}

\newcommand{\Res}[0]{\mathrm{Res}}
\newcommand{\rank}[0]{\mathrm{rank}}
\newcommand{\reg}[0]{\mathrm{reg}}

\newcommand{\inv}[0]{^{-1}}

\newcommand{\ts}{\textsuperscript}

\newcommand{\MHiggs}[0]{{\mathcal{M}^{\textrm{Higgs}}}}
\newcommand{\Mtoy}[0]{ \calM^{\textrm{toy }}}
\newcommand{\Mell}[0]{ \calM^{\textrm{ell }}}
\newcommand{\bloch}[0]{\beta}

\newcommand{\HH}[0]{\mathbb{H}}
\newcommand{\calH}[0]{\mathcal{H}}
\newcommand{\calM}[0]{\mathcal{M}}

\newcommand{\calN}[0]{\mathcal{N}}
\newcommand{\calB}[0]{\mathcal{B}}

\newcommand{\calO}[0]{\mathcal{O}}
\newcommand{\RR}[0]{\mathbb{R}}
\newcommand{\CC}[0]{\mathbb{C}}
\newcommand{\ZZ}[0]{\mathbb{Z}}
\newcommand{\PP}[0]{\mathbb{P}}

\newtheorem{proposition}{Proposition}

\newtheorem{pipeDream}{Daydream}
\newtheorem{theorem}{Theorem}

\theoremstyle{problem}
\newtheorem{problem}{Problem}

\theoremstyle{definition}
\newtheorem{definition}{Definition}

\theoremstyle{remark}
\newtheorem{remark}{Remark}

\title{Hyperbolic band theory through Higgs Bundles}
\author{Elliot Kienzle}
\address{Department of Mathematics, University of Maryland, William E. Kirwan Hall, 4176 Campus Drive, College Park, MD, USA 20742-4015}
\email{elliot.kienzle@gmail.com}

\author{Steven Rayan}
\address{Centre for Quantum Topology and Its Applications (quanTA) and Department of Mathematics \& Statistics, University of Saskatchewan, McLean Hall, 106 Wiggins Road, Saskatoon, SK, CANADA S7N 5E6}
\email{rayan@math.usask.ca}

\date{\today}

\begin{document}

\begin{abstract}Hyperbolic lattices underlie a new form of quantum matter with potential applications to quantum computing and simulation and which, to date, have been engineered artificially. A corresponding hyperbolic band theory has emerged, extending $2$-dimensional Euclidean band theory in a natural way to higher-genus configuration spaces. Attempts to develop the hyperbolic analogue of Bloch's theorem have revealed an intrinsic role for algebro-geometric moduli spaces, notably those of stable bundles on a curve.  We expand this picture to include Higgs bundles, which enjoy natural interpretations in the context of band theory. First, their spectral data encodes a crystal lattice and momentum, providing a framework for symmetric hyperbolic crystals. Second, they act as a complex analogue of crystal momentum. As an application, we elicit a new perspective on Euclidean band theory. Finally, we speculate on potential interactions of hyperbolic band theory, facilitated by Higgs bundles, with other themes in mathematics and physics.
\end{abstract}

\maketitle

\tableofcontents

\section{Introduction}

There are a number of established and emerging connections between high energy and condensed matter physics, some of which are mediated by physical and mathematical dualities such as AdS/CFT.  This article aims to describe one such emerging connection, specifically between the moduli space of Higgs bundles, a feature of Yang-Mills gauge theory in $2$ dimensions, and the band theory of media with periodic potentials.  The former originates in \cite{hitchin_self-duality_1987}, where the self-dual Yang-Mills equations are dimensionally reduced so that they are written on a smooth Hermitian bundle on a compact Riemann surface.  For the latter, we take the point of view not only of the familiar $2$-dimensional Euclidean band theory given by the symmetry of a lattice tiling the plane, but also expand our attention to include its hyperbolic generalization that has recently appeared in \cite{maciejko_hyperbolic_2021}.  The hyperbolic version of the theory considers the Schr\"odinger equation with a Hamiltonian that is invariant under the noncommutative translations of a tessellation of the hyperbolic plane $\HH$.  One of the great successes of conventional band theory is the explanation of electronic and optical properties of solids, which allows for materials to be realized with desired semiconducting or insulating properties through band-gap engineering and which provides the theoretical foundation for solid-state devices such as transistors.  The hyperbolic realization of band-gap phenomena anticipates new forms of quantum matter and topoelectric circuits with applications to quantum computing and simulation.  To date, such circuits, albeit photonic rather than electronic in nature, have been artificially engineered \cite{kollar_hyperbolic_2019}.  A subsequent experiment realizes a hyperbolic drum as an electric circuit \cite{drum_2021}, culminating in the experimental detection of negative curvature via the detection of the eigenmodes of the Laplace-Beltrami operator and of signal propagation along hyperbolic geodesics.  Further theoretical works emerging recently in hyperbolic matter and band theory include \cite{ikeda_hyperbolic_2021,aoki_algebra_2021,boettcher_crystallography_2021,stegmaier_universality_2021,attar_selberg_2022,bienias_circuit_2022,boettcher_insulators_2022}.

One of the reasons why the hyperbolic generalization of band theory is attractive is because it continues to admit Bloch wave solutions, as established in \cite{maciejko_hyperbolic_2021}.  The Bloch decomposition here is automorphic in nature, with phase factors given by unitary representations (in all ranks) of the Fuchsian group $\Gamma$ of the tessellation.  The appearance of representations of surface groups of higher-genus curves $\HH/\Gamma$ is a natural entry point for Higgs bundles into the band theory.

After reviewing the hyperbolic band theory as well as the definition and basic features of Higgs bundles and their moduli, we connect the two subjects in two ways.  First, we package the crystal lattice and abelian crystal momenta of the band theory into the spectral data of Higgs bundles, which thereby injects the base of the Hitchin map on the moduli space of Higgs bundles into the space of available crystals.  Alternatively, we interpret the Higgs field of a Higgs bundle as an imaginary crystal momentum state, which leverages the similarity between spectral curves and band structures within a tight-binding model.  Even in the Euclidean case, these interpretations elicit new aspects of the band theory, which we will illustrate with some suggestive calculations for parabolic Higgs bundles on the projective line with elliptic spectral curves.  In the final section of the paper, we speculate on connections to topological materials (which are in some sense the major impetus for the recent injection of algebraic techniques into condensed matter theory), to fractional quantum Hall states, to supersymmetric field theories, to the geometric Langlands correspondence, and other themes.  We end the article with a rather provocative picture that seeks to place a number of ideas from high-energy physics, algebraic geometry, and condensed matter into the same orbit.

The starting point of condensed matter physics is crystallography and the theory of a quantum particle on a periodic Euclidean background. The theory is set in motion by a Hamiltonian with a potential that periodic in space. Mathematically, we organize particles propagrating under the potential into irreducible representations of the symmetry group, a discrete subgroup of translations of Euclidean space. These are classified by a point in reciprocal space known as the crystal momentum, resulting in the famous Bloch's theorem. Such a general outlook opens the possibility for crystals in other symmetric spaces, where the symmetries of the potential form a discrete subgroup of isometries. In this paper, we ask about hyperbolic crystals, periodic in the $2$-dimensional hyperbolic plane $\HH$ which carries a metric of constant negative curvature. The isometry group here is $PSL(2,\RR)$, with Fuschian groups as the discrete subgroups. The orbit of a point under such a group describes the crystal lattice as a hyperbolic lattice in $\HH$. Following the principle above, the hyperbolic analogue of Bloch's theorem observes that quantum states in a hyperbolic crystal are classified by an irreducible representation of the associated Fuschian group $\Gamma$. Unlike Euclidean crystals, $PSL(2,\RR)$ and thus the Fuschian group are nonabelian, and so the irreducible representations of $\Gamma$ now form an interesting moduli space.

So far, the problem is confined to quantum mechanics, dealing exclusively with irreducible unitary representations of a symmetry group. We will approach this by converting the global, algebraic data to geometric data on a generating set of the symmetry. This picture consists of the fundamental polygon with edges identified according to lattice translations, giving a compact $2$-dimensional surface $\Sigma = \HH / \Gamma$ with an inherited metric of constant negative curvature. This defines a unique compact Riemann surface. We assume $\Gamma$ acts without fixed points, and so the fundamental group $\pi_1(\Sigma)$ is isomorphic to $\Gamma$. The Riemann-Hilbert correspondence exchanges each representation of the fundamental group for a unique flat connection on a smooth vector bundle over $\Sigma$, and the kinetic term of the Hamiltonian becomes the Laplacian for the flat connection. The flat connection also corresponds to a holomorphic structure, as part of the correspondence between points in the moduli space of unitary representations and points in the moduli space of stable holomorphic vector bundles. Hyperbolic band theory seeks the spectrum of the Hamiltonian for each crystal momentum, which is now a problem in differential and complex algebraic geometry.

This paper integrates hyperbolic band theory with the study of Higgs bundles, which arise in two seperate but equally natural contexts in the band theory. First, there are crystals with symmetry. A crystal lattice with spatial inversion symmetry corresponds to a Riemann surface $\Sigma$ with involution $\sigma$. The minimal generating set is then half of the fundamental polygon, whose edges glue together to the quotient $\Sigma/\sigma = \PP^1$ --- that is, $\Sigma$ is a hyperelliptic curve.  An abelian (i.e. rank-$1$) crystal momentum defines a holomorphic line bundle on $\Sigma$. We can represent this roughly a matrix at each point whose graph of eigenvalues defines $\Sigma$ and associated eigenvectors define the line bundle. More formally, this is a pair of a rank 2 holomorphic vector bundle $E$ and a holomorphic endomorphism valued 1-form $\phi$, called a Higgs bundle. This exemplifies the procedure for encoding a crystal as the spectral curve of a parabolic Higgs bundle, and an Abelian crystal momentum as the spectral line bundle. The spectrum of the crystal Hamiltonian on the spectral curve pushes forward to an analogous Hamiltonian on the base curve, replacing the flat connection with the parabolic, higher rank pushforward connection (Theorem \ref{thm:nonabelian hamiltonain}).

By encoding crystal data in Higgs bundles, the moduli space of Higgs bundles becomes a moduli space of crystal data. The Hitchin fibration naturally splits this into the crystal lattice (Hitchin base) and abelian crystal momenta (fibers, or the Jacobian of the spectral curve). Each Higgs bundle has an associated crystal Hamiltonian, defining a band structure over the whole moduli space of Higgs bundles. Essentially, the band structure over the Jacobian of each crystal glues together into a global band structure. This suggestively hints at a universal band structure, framing hyperbolic band theory as a moduli problem.

Higgs bundles also naturally arise as a complex crystal momentum. In hyperbolic band theory, the space of crystal momenta is the variety of irreducible unitary representations $\pi_1(\Sigma)\to U(n)$, which the Narisimhan-Seshardi theorem equates with the moduli space of stable vector bundles. We complexify the problem by replacing $U(n)$ with $GL(n,\CC)$. The nonabelian Hodge theorem equates the variety of irreps $\pi_1(\Sigma)\to GL(n,\CC)$ to the moduli space of stable \textit{Higgs} bundles. The physical interpretation is clearest in the abelian case. The representation $\pi_1(\Sigma) \to GL(1,\CC)\cong \CC^*$ give the factors of automorphy for an electron with that crystal momentum. If the image lies in $U(1)$, the electron preserves its norm from cell to cell, corresponding to a purely real momentum. If not, the electron gains or loses amplitude across the crystal, corresponding to a complex momentum. This gives an associated non-Hermitian Hamiltonian with complex eigenstates. All in all, the two interpretations of Higgs bundles define two separate band structures over the moduli space of Higgs bundles.

This dictionary between hyperbolic crystals and Higgs bundles motivates a number of speculations.  For instance, topological materials --- a major impetus for the recent injection of algebraic topology into condensed matter theory --- are classified by a number of topological invariants.  The large dimension of the space of hyperbolic crystal momenta permits more exotic topological invariants than in Euclidean band theory.  More generally, Higgs bundles have come to enjoy a role as an intermediary between various objects in geometry, representation theory, and physics. As such, the connection between Higgs bundles and hyperbolic band theory suggests band-theoretic interpretations of these objects.  Our speculations involve fractional quantum Hall states, supersymmetric field theories, the geometric Langlands correspondence, and other ideas. We end the article with a rather provocative picture that aims to put many ideas in high-energy physics, algebraic geometry, and condensed matter into the same orbit.

In regards to the precise organization of the paper, Section 2 synthesizes the hyperbolic band theory framework of \cite{maciejko_hyperbolic_2021,maciejko_automorphic_2021} in a self-contained and mathematically-oriented manner. It contributes a description of the abelian and nonabelian crystal Hamiltonian, with various geometric and physical interpretations. In Section 3, we review Higgs bundles and their moduli spaces. We combine these in sections 4 and 5, respectively describing the crystal moduli and complex momenta interpretation of Higgs bundles. Section 6 tests this framework in the realm of Euclidean crystals, represented by the toy model of rank-$2$ Higgs bundles on $\PP^1$ with $4$ parabolic points. This wraps the standard band theory canon into an algebro-geometric package. We conclude in Section 7 by speculating on the place of hyperbolic band theory in a larger story of mathematics and physics, as facilitated by Higgs bundles.\\

\noindent\emph{Acknowledgements.} We gratefully acknowledge Igor Boettcher, Kazuki Ikeda, Joseph Maciejko, Rafe Mazzeo, \'Akos Nagy, Robert-Jan Slager, Jacek Szmigielski, and Richard Wentworth for very helpful discussions.  The second-named author is partially supported by an NSERC Discovery Grant, a Tri-Agency New Frontiers in Research Fund (Exploration Stream) award, and a Pacific Institute for the Mathematical Sciences (PIMS) Collaborative Research Group grant.  The first-named author was supported by the NSERC Discovery Grant of the second-named author.  The illustrations in this manuscript were created by the first-named author.

\section{Hyperbolic band theory} \label{sec:hyperbolic_Band_Theory}

We begin by building the framework of hyperbolic band theory. The foundation is a hyperbolic version of the classical Bloch theorem, which constructs the eigenstates of the crystal Hamiltonian from crystal momenta. Hyperbolic Bloch states were introduced in \cite{maciejko_hyperbolic_2021}, and a corresponding version of Bloch's theorem was formulated in \cite{maciejko_automorphic_2021}.  
\subsection{Euclidean crystallography}\label{sec:Euclidian_Bloch}

Hyperbolic crystallography is modeled on Euclidean crystallography, which we review for inspiration. This material is standard (see for instance \cite{ashcroft_solid_1976}) but presented here with the hyperbolic generalization in mind. A crystal is defined by a periodic potential $V(x)$ on $\RR^n$, invariant under the crystal lattice $\Gamma \sub \RR^n$. Invariance means $V(x) = V(x+\gamma) $ for all points $\gamma$ in $\Gamma$. The crystal lattice is described by its symmetry, treating $\Gamma$ as a discrete subgroup of translations in $\RR^n$. Finally, the lattice is uniquely determined by the shape of its unit cell, a polyhedral set whose translates under $\Gamma$ tile all of $\RR^n$. Periodicity identifies opposite faces of the cell, which glue together to form the torus $\RR^n/\Gamma$, as summarized for two dimensions in Figure \ref{fig:Euclidean_unit_cell}. Periodic functions on $\RR^n$ are uniquely described by ordinary functions on $\RR^n/\Gamma$. We are interested in hyperbolic crystals in $2$ dimensions, which are analogous to Euclidean crystals in $\RR^2 \cong \CC$.

\begin{figure}[htp]
    \centering
    \includegraphics[width = \textwidth]{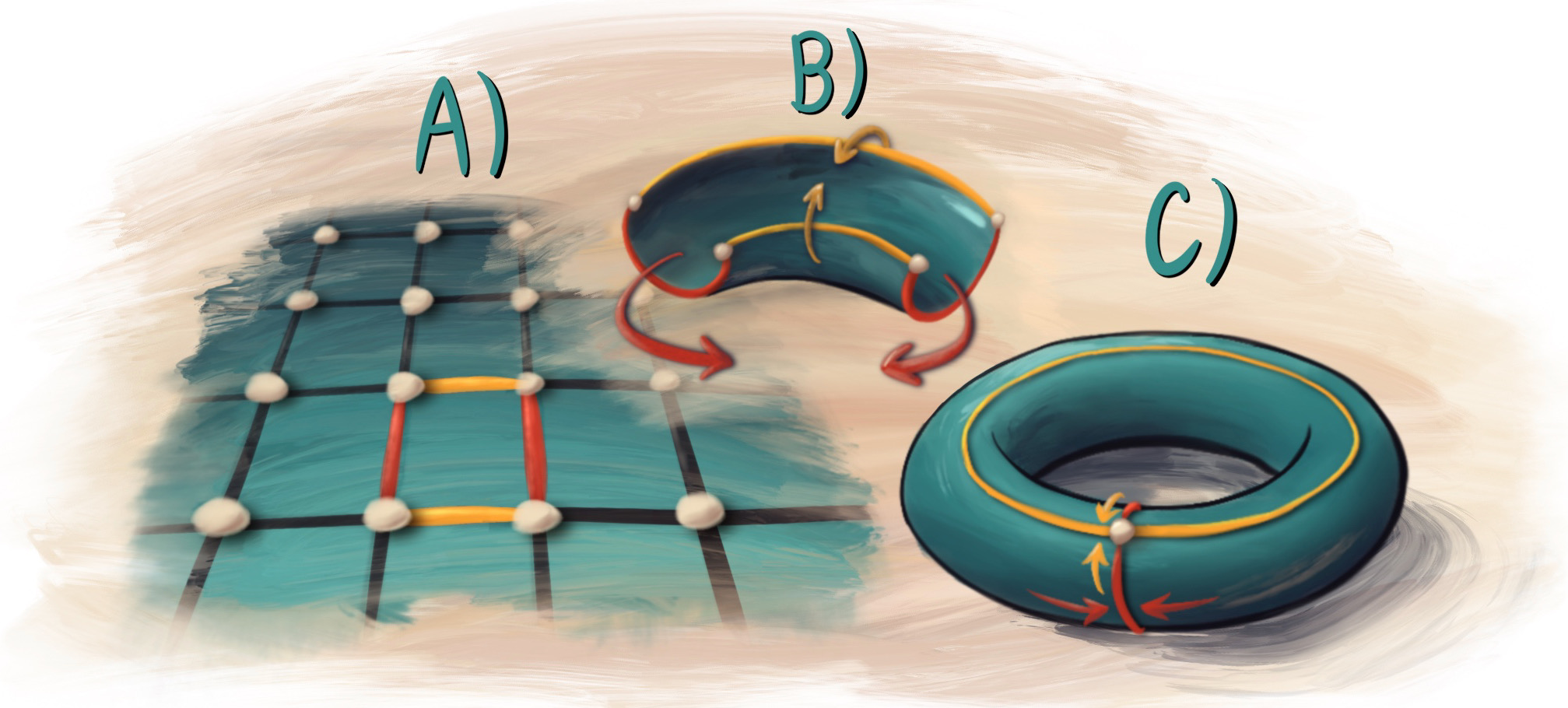}
    \caption{A) A Euclidean crystal, given by a lattice $\Gamma$ in $\CC$, with a unit cell highlighted. B) to construct $\CC/\Gamma$, the edges of the unit cell of like color are identified. C) The resulting surface is a torus. }
    \label{fig:Euclidean_unit_cell}
\end{figure}

\subsubsection{Bloch's theorem} \label{sec: Euclidian bloch thm}
Consider a quantum particle on this periodic potential. Its wavefunction is a complex-valued function $\psi: \CC \to \CC$, whose evolution is governed by the Hamiltonian $H = \Delta + V$. Here $\Delta$ is the Laplacian and $V: \CC \to \RR$ is the potential energy. Band theory is founded on \textit{Bloch's theorem}, which says the Hamiltonian has an eigenbasis $\psi_k$ satisfying $\psi_k(z + \gamma) = e^{2 \pi i (k_x \gamma_x + k_y \gamma_y)}\psi_k(z)$, which are called Bloch waves with crystal momentum $k = k_x + i k_y \in \CC$. 

A generalizable derivation of Bloch's theorem comes from group theory. To summarize, the Hilbert space $\calH = L^2(\CC)$ splits into irreducible representations of $\Gamma$ parametrized by $k$, each representation containing Bloch waves with crystal momentum $k$. The Hamiltonian is symmetric under lattice translations, and so it preserves these subspaces. Each subspace has an energy eigenbasis, which together forms a complete eigenbasis for $H$. More carefully, we have an action of $\Gamma$ on the Hilbert space $\calH = L^2(\CC)$ via translation operators $T_\gamma \psi(z) = \psi(z+\gamma)$. These form a representation of $\Gamma$ because $T_{\gamma_1 + \gamma_2} = T_{\gamma_1} \cdot T_{\gamma_2}$. The Hilbert space $\calH$ splits into a direct product of irreducible representations of $\Gamma$, all of which are $1$-dimensional since $\Gamma$ is abelian. Denoting the space of these representations by $\Gamma^\vee$, we have

\[
    \calH = \prod_{\chi_k \in \Gamma^\vee} \calH_k,  \qquad \calH_k = \{\psi \in \calH | T_\gamma \psi = \chi_{k}(\gamma) \psi\}.
\]

The representations are defined by $\chi_k(\gamma) = e^{2 \pi i  k^*(\gamma)}$, parametrized by a linear function $k^*: \CC \to \CC$, the space of which is isomorphic to $\CC$. This is trivial whenever $k^*(\gamma)$ is an integer for every $\gamma\in \Gamma$, which by definition means $k^*$ lies in the dual lattice $\Gamma^* \subset \CC$. We are left with the well-known isomorphism $\Gamma^\vee \cong \CC/\Gamma^*$, identifying the space of irreducible representations as a $1$-complex dimensional torus. We can expand out the action of $k^*$ via $k^*(\gamma) =  k\cdot \gamma = k_x \gamma_x + k_y \gamma_y$. So, $\calH_k$ exactly consists of Bloch waves with crystal momentum $k$. The identification $\Gamma^\vee \cong \CC/\Gamma^*$ reflects the periodicity of crystal momentum space.

To connect this splitting with the Hamiltonian, note that translation by $\gamma$ is an isometry, and so $T_\gamma$ is unitary and commutes with the Laplacian. As the potential is periodic, it also commutes with lattice translations, meaning the Hamiltonian satisfies $T_\gamma H = H T_\gamma$. Hence, $H$ preserves every subspace $\calH_k$. Denoting the restriction of $H$ to $\calH_k$ by $H_k$, we arrive at: 

\begin{theorem} [Bloch's theorem]
A periodic Hamiltonian $H$ with lattice $\Gamma \sub \CC$ has an eigenbasis of Bloch waves. This basis consists of eigenfunctions of $H_k = H | _{\calH_k}$ for $k$ varying in crystal momentum space $\CC^\times / \Gamma^*$.
\end{theorem}

Bloch's theorem is central to band theory, which studies the spectrum of $H_k$ as $k$ varies across crystal momentum space.

\begin{remark}
We should be clear that this derivation is not rigorous as given. The Hilbert space is not $L^2(\CC)$, as any nonzero quasiperiodic function is not square-integrable. Instead, this is a physical argument motivating the splitting of $H$ into spaces of Bloch waves. 
\end{remark}

\subsection{Hyperbolic crystallography}

\begin{figure}[htp]
    \centering
    \includegraphics[width = \textwidth] {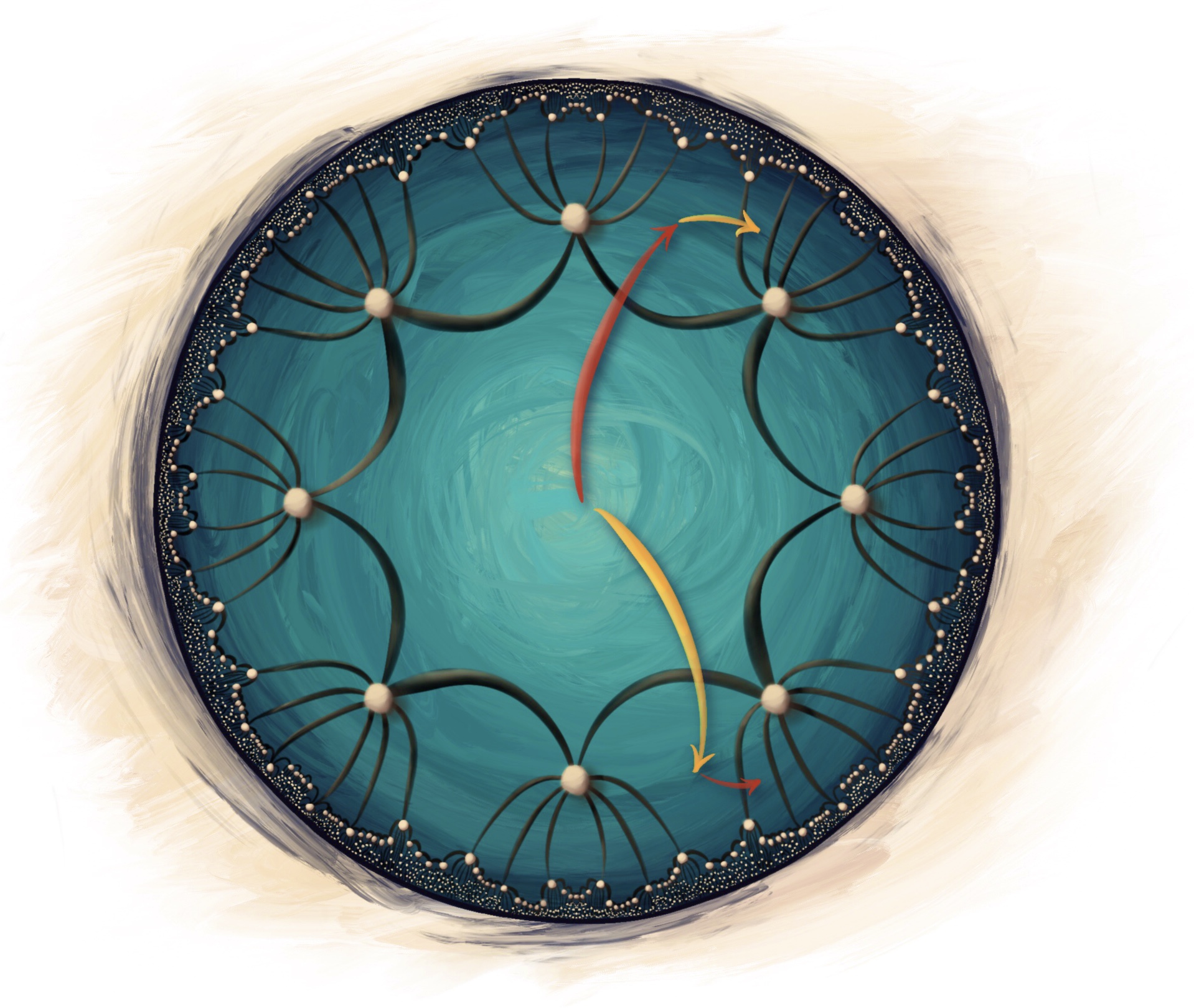}
    \caption{  A hyperbolic lattice in the Poincar\'e disk model. The white dots show the orbit of a point under a Fuchsian group, and the hyperbolic octagon in the center of the image depicts the unit cell. Red and yellow arrows show the action of two elements of the Fuchsian group. The final position depends on the order of the red and yellow actions, reflecting the noncommutativity of the Fuchsian group.}
    \label{fig:Nonableian }
\end{figure}

This paper studies hyperbolic crystals, or periodic structures on the hyperbolic plane $\HH$. As usual, $\HH$ is a $2$-dimensional open disk endowed with a metric of constant negative curvature, whose isometries form its symmetry group $PSL(2,\RR)$. A crystal structure is described by a discrete subgroup of isometries $\Gamma \sub PSL(2,\RR)$ (a Fuchsian group) acting on $\HH$ without fixed points. A potential $V:\HH \to \RR$ is periodic with crystal structure $\Gamma$ when $V(x) = V(\gamma(x))$ for all hyperbolic translations $\gamma$ in $\Gamma$. Much like Euclidean crystals, we visualize this through a Hyperbolic lattice consisting of the orbit of one point under $\Gamma$. Figure \ref{fig:Nonableian } depicts this in the Poincar\'e disk model. The crystal structure is also determined by its unit cell, consisting of a hyperbolic polygon bounded by geodesics and depicted in Figure \ref{fig:Nonableian } with black lines. The elements of $\Gamma$ move from one crystal cell to the next, indicated by the yellow and red arrows. Periodicity identifies opposite sides of a unit cell, forming a 2D surface $\HH/\Gamma$, shown pictorially in Figure \ref{fig:Hyperbolic_unit_cell}. The surface inherits the constant negative curvature metric from $\HH$, defining a Riemann surface. The uniformization theorem states that every Riemann surface $\Sigma$ equals $\HH/\Gamma$ for some Fuchsian group $\Gamma$. 

\begin{figure}[htp]
    \centering
    \includegraphics[width = \textwidth]{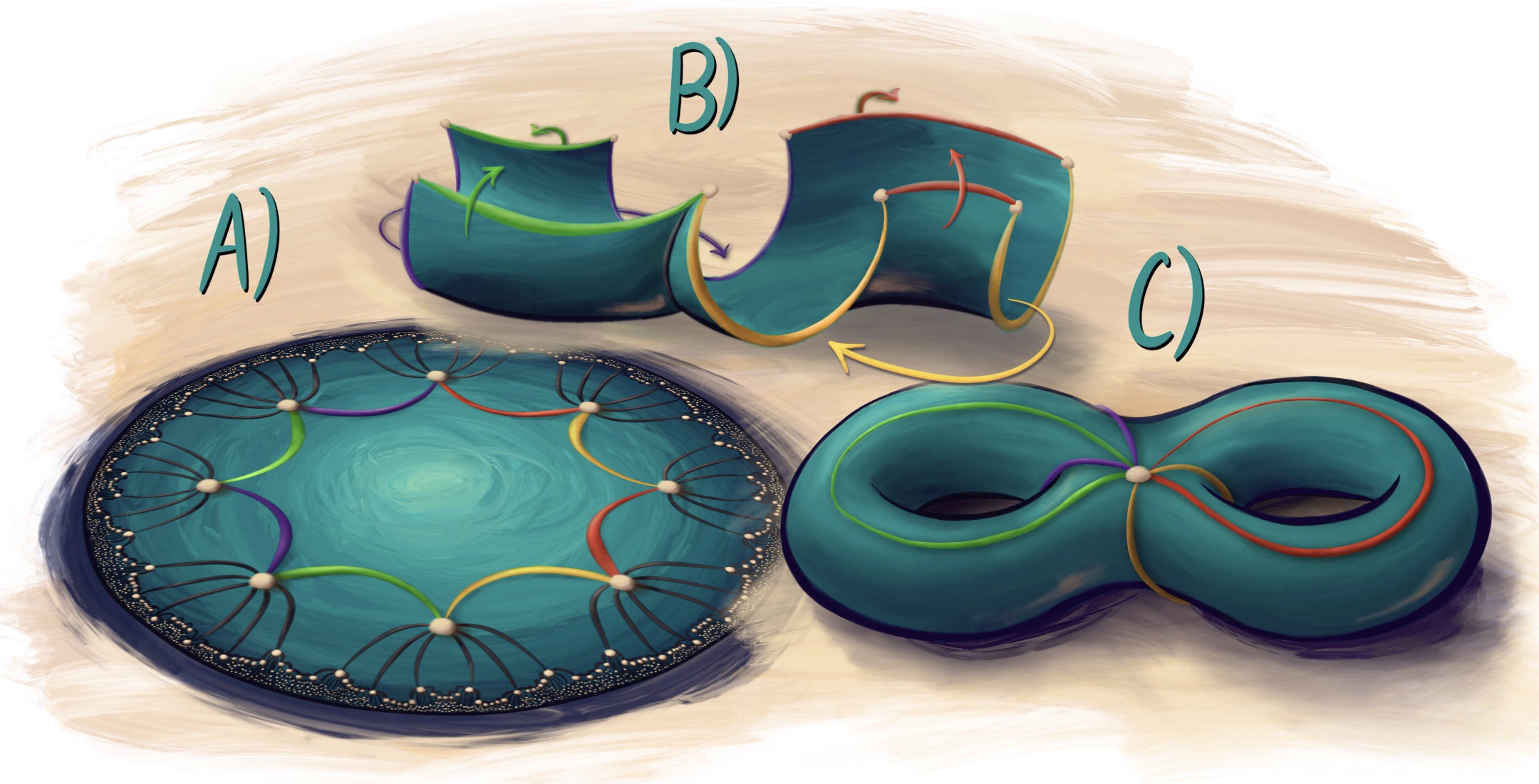}
    \caption{A) A hyperbolic lattice in the hyperbolic plane, with unit cell highlighted. For the genus two lattice pictured, the unit cell has $4g=8$ sides. B) To construct $\HH/\Gamma$, the opposite edges of the unit cell are identified. This figure shows an alternate identification scheme, where the edges are labeled $aba\inv b\inv c d c\inv d\inv$.  C) This results in a genus $g$ surface. }
    \label{fig:Hyperbolic_unit_cell}
\end{figure}

\subsubsection{Hyperbolic Bloch theorem}
Just as conventional band theory starts with Bloch's theorem, hyperbolic band theory is built on a hyperbolic analogue of Bloch's theorem. The difficulty in extending Section \ref{sec: Euclidian bloch thm} to hyperbolic crystals is the noncommutativity of the lattice, meaning irreducible representations can be larger in dimension than $1$.

For a crystal with symmetries defined by a Fuschian group $\Gamma \sub PSL(2,\RR)$, the crystal potential must be invariant under the action of $\Gamma$. Just as before, $\Gamma$ acts on Hilbert space $\calH = L^2(\HH)$ by pre-composition $T_\gamma \psi(z) = \psi(\gamma(z))$, furnishing a representation of $\Gamma$. Since $\Gamma$ acts by isometries, $T_\gamma$ is unitary and commutes with both the potential and the Laplacian. We again split $\calH$ into irreducible representations of $\Gamma$, but since $\Gamma$ is noncommutative, we must deal with the technical challenges of higher-dimensional irreducible representations.

Consider a unitary representation $\rho: \Gamma \to U(n)$. We say a state in $\calH$ transforms according to $\rho$ if it belongs to an n-dimensional subspace $\langle \psi_1, \dots , \psi_n \rangle $ where each state $v_1\psi_1 + \dots + v_n\psi_n$ satisfies
\[
    T_\gamma \begin{pmatrix}\psi_1 & \hdots & \psi_n\end{pmatrix} \begin{pmatrix}v_1 \\ \vdots \\ v_n\end{pmatrix} = \begin{pmatrix}\psi_1 & \hdots & \psi_n\end{pmatrix} \rho(\gamma)\begin{pmatrix}v_1 \\ \vdots \\ v_n\end{pmatrix}.  
\]

In physics parlance, we say the state belongs to a multiplet of $\Gamma$. $\calH$ splits into a irreducible representations of $\Gamma$. Collect all the states transforming according to $\rho$ into a subspace $\calH_\rho \sub \calH$. Then, $\calH$ is a direct product of all $\calH_\rho$ as $\rho$ varies over the space of irreducible representations. Given that $H$ commutes with each $T_\gamma$, we check when $H$ preserves $\calH_\rho$:
\[
 T_\gamma \begin{pmatrix}H \psi_1 & \hdots & H\psi_n\end{pmatrix} \begin{pmatrix}v_1 \\ \vdots \\ v_n\end{pmatrix} =
 H T_\gamma \begin{pmatrix}\psi_1 & \hdots & \psi_n\end{pmatrix} \begin{pmatrix}v_1 \\ \vdots \\ v_n\end{pmatrix} = 
 \begin{pmatrix}H\psi_1 & \hdots & H\psi_n\end{pmatrix} \rho(\gamma) \begin{pmatrix}v_1 \\ \vdots \\ v_n\end{pmatrix}.
\]

When the vectors $\{H \psi_i\}$ are linearly independent, they form a representation of $\rho$, so each vector $H \psi_i$ lives in $\calH_\rho$. We are not so lucky when $\{H \psi_i\}$ is linearly dependent, as we instead obtain a lower-dimensional representation, meaning $H\psi_i$ does not live in $\calH_\rho$. This can only happen when $H$ has a nontrivial kernel. When the kernel is trivial, $H$ preserves $\calH_\rho$ for all $\rho$, so it restricts to a self-adjoint operator $H_\rho$ giving us: 

\begin{theorem}[Hyperbolic Bloch theorem]
When the Hamiltonian $H$ has trivial kernel,  $H$ has an eigenbasis of hyperbolic Bloch waves (states belonging to $\calH_\rho$ for some irreducible unitary representation $\rho$.)
These consist of eigenfunctions of $H_\rho = H|_{\calH_\rho}$.

\end{theorem}

This argument is essentially given in \cite{maciejko_automorphic_2021}, though framed slightly differently. That paper dealt with a finite subset of the full hyperbolic lattice in the tight-binding model, and consequently a finite-dimensional Hilbert space. This avoids several subtle issues with the above argument, like the lack of a proper Hilbert space and the assertion that all representations are finite-dimensional. Once again, we forego a rigorous treatment of the preceding argument and take it as a physical heuristic about what hyperbolic crystal momentum space ought to be. 

This momentum space is the moduli space of  irreducible unitary representations of $\Gamma$, $\Hom_\textrm{irr}(\Gamma \to U(n)) /U(n)$. We mod out by the conjugation action of $U(n)$ because two conjugate crystal momenta are equivalent. The most familiar component consists of rank one representations $\Gamma^\vee \equiv \Hom(\Gamma,U(1))$.  Since $U(1)$ is abelian, any homomorphism $\Gamma \to U(1)$ factors through the abelianization of $\Gamma$, which is isomorphic to $\ZZ^{2g}$ for an integer $g\geq 0$. Therefore, we have \[\Gamma^\vee \cong \Hom(\ZZ^{2g},U(1)) \cong U(1)^{2g}.\]
This component is a complex $g$-dimensional torus, thereby generalizing the Euclidean case, in which $g=1$ and where the momentum space is a complex $1$-dimensional torus. We call this the \textit{abelian Brillouin zone}, which was introduced for hyperbolic crystals in \cite{maciejko_hyperbolic_2021}. 
The other connected components of momentum space are the character varieties  $\Hom_\irr(\Gamma,U(n))/U(n)$, which we will call the \textit{nonabelian Brillouin zones}, introduced in \cite{maciejko_automorphic_2021}. We can express these geometrically as moduli spaces of vector bundles on Riemann surfaces.

\subsection{Vector bundles on Riemann surfaces} \label{sec:Riemann Surfaces}

The algebraic description of crystal momentum space arose by thinking of a crystal through its symmetry group. Geometry comes from the unit cell picture, representing the crystal with a Riemann surface $\Sigma = \HH / \Gamma$. Since $\Gamma$ acts freely and properly on the simply connected space $\HH$, the fundamental group $\pi_1(\Sigma)$ is isomorphic to $\Gamma$. Crystal momenta are representations of the fundamental group, which all arise from the monodromy of a flat connection, due to the Riemann-Hilbert correspondence. 

This is clearest for abelian Bloch states, parametrized by $U(1)$ representations of the fundamental group \cite{maciejko_hyperbolic_2021}. We start with a trivial complex line bundle over $\HH$. A Bloch state $\psi$ with crystal momentum $k$ is a section of this bundle with factor of automorphy $\chi_k$, which we can represent as a section of some line bundle $L$ over $\Sigma$. Specifically, $L$ is constructed from the trivial line bundle $\HH \times \CC$ by quotienting by $\Gamma \cong \pi_1(\Sigma)$, acting via $\gamma: (x,v) \to (\gamma(x),\chi_k(\gamma) (v)$. In effect, the fibers rotate by $\chi_k(\gamma)$ over any loop $\gamma$. We encode this as the parallel transport of a flat $U(1)$ connection $\nabla_L$, unique up to gauge transform. Note that $L$ has a Hermitian metric, pulled back from the standard Hermitian metric on the trivial bundle $\HH \times \CC$, $h(x,y) = x\bar{y}$. We can regard $\nabla_L$ as a Hermitian $\CC^*$ connection, which reduces the structure group from $\CC^*$ to $U(1)$.

Moving to complex geometry, the $(0,1)$ part of $\nabla_L$ gives a Dolbeault operator $\delbar_L$, endowing $L$ with a holomorphic structure. Through this map, the abelian Brillouin zone $\Hom(\Gamma, U(1))$ is isomorphic to the moduli space of degree 0\footnote{The degree is $0$ because the line bundle is topologically trivial.}
holomorphic line bundles on $\Sigma$, better known as the Jacobian variety $\Jac(\Sigma)$. For $\Sigma$ a genus $g$ surface, this is a $g$-dimensional complex torus, agreeing with the last section. For $g=1$, the Jacobian is the dual abelian variety to the position-space torus, reproducing the duality between momentum and position space for Euclidean crystals.

A similar scene plays out for higher rank representations, but the details are more involved \cite{maciejko_automorphic_2021}. Starting from a trivial vector bundle $\HH \times \CC^n \to \HH$, we obtain a topologically trivial\footnote{The topological type of a vector bundle is determined by its Chern classes, and on a Riemann surface the only nonzero Chern class is $c_1$. So, a vector bundle is topologically trivial if and only if it has degree $0$.}
vector bundle $E \to \Sigma$ by quotienting out $\Gamma$ acting by $\gamma: (x,v) \to (\gamma(x),\rho(\gamma) (v))$. Pulling back the standard Hermitian inner product on $\HH \times \CC^n$ gives a flat Hermitian vector bundle $(E,h)$ over $\Sigma$. We can associate the crystal momentum $\pi_1(\Sigma) \to U(n)$ to a flat $U(n)$ connection $\nabla_E$. In the algebraic geometry realm, this gives $E$ a holomorphic structure $\delbar_E = \nabla_E^{0,1}$. The Narasimhan-Seshadri theorem \cite{narasimhan_stable_1965} says that a holomorphic bundle arises from an irreducible representation if and only if it is stable, which is a constraint on the (normalized) degrees of holomorphic subbundles. Specifically, all subbundles $V\sub E$ must satisfy 
\[\frac{\deg(V)}{\rank(V)} < \frac{\deg(E)}{\rank(E)}.\]
That is, $E$ has the maximal normalized degree of all its subbundles. Stability ensures that the automorphism group of a vector bundle is as small as possible, which is necessary for a well-behaved (i.e. Hausdorff) moduli space. The nonabelian Brillouin zone $\Hom_\irr(\Gamma,U(n))/U(n)$ is diffeomorphic to the moduli space of degree zero stable vector bundles $\calN^s(\Sigma,n)$. This is a smooth manifold with real dimension $2n^2(g-1) + 2$. Notably, it is no longer a torus when $n>1$, instead developing a complicated but well-studied topology. It is noncompact, which is remedied by including semistable bundles (replacing $<$ in the definition of stability with $\leq$), which come from irreducible representation. The application of these moduli spaces to hyperbolic band theory is discussed in \cite{maciejko_automorphic_2021}. 

\subsection{Crystal Hamiltonian for abelian Bloch states}

The hyperbolic analogue of Bloch's theorem initiates hyperbolic band theory. This studies the spectrum of the crystal Hamiltonians $H_\rho$ as $\rho$ varies across the space of crystal momenta. We approach this with a geometric formulation of $H_\rho$. To start, we consider abelian crystal momenta defined by a representation $\chi_k:\Gamma \to U(1)$ with Hamiltonian $H_k$. Finding the spectrum of $H_k$  amounts to solving an eigenvalue problem for $H$ on the unit cell with twisted periodic boundary conditions  $\psi(\gamma(x)) = \chi(\gamma) \psi(x)$. This was studied numerically using a finite element method in reference \cite{maciejko_hyperbolic_2021}. We seek a closed form of $H_k$ which is more amenable for analytic manipulation.

\begin{theorem}\label{thm:momentum}
Let $A$ be a closed  one-form satisfying $\exp(\int_\gamma A )= \chi(\gamma)$ for every loop $\gamma$. Then, every eigenfunction of $H_k$ corresponds uniquely to an eigenfunciton with equal eigenvalue of the Hamiltonian 
\begin{equation} \label{eq:momentum_hamiltonian}
    H_L = (d + A)^* (d+A) + V
\end{equation}
acting on $L^2(\Sigma)$. In particular, the spectrum of $H_k$ and $H_L$ are equal.\footnote{The Hamiltonian defined on $\Sigma$ is denoted $H_L$ because $A$ will later be associated to a flat line bundle $L$ over $\Sigma$.}
\end{theorem}

We will pass between this coordinate invariant form and its expression in local coordinates. Choose a conformal coordinate $z$ on a patch of $\Sigma$, where the metric has the form $g(z) \de z ^2$ and $A = A'(z)\de z + A'' (z) \de \zbar$.  The Hamiltonian is then

\begin{equation} \label{eq:Local form of abelian hamiltonian}
    H_L = \frac{1}{g} (\delz+A') (\delbarz+A'') + V. 
\end{equation}

\begin{proof}[Proof: Theorem \ref{thm:momentum}]
We once again follow the Euclidean derivation for guidance. In that case, we can express any Bloch wave as $\psi_k(z) = e^{ik\cdot z} u(z)$, where $k,z$ are points in $\CC$ thought of as vectors in $\RR^2$, and $u(z)$ is periodic with the lattice $\Gamma$. In effect, we split $\psi_k$ into a predetermined phase factor $e^{ik\cdot z}$ satisfying proper boundary conditions, and a periodic remainder $u(z)$.  Denoting $k_z = k_x + i k_y, k_{\zbar} = k_x - i k_y$, and noting that $g=1$ on the torus, we have

\[
    H \psi_k =  ( \delz \delbarz + V) e^{ik\cdot z} u = e^{ik\cdot z} ((\delz + ik_z)(\delbarz + i k_{\zbar}) + V) u.
\]
Hence, if $\psi_k$ is an eigenfunction of $H$ with eigenvalue $\lambda$, then $u$ is an eigenfunction of 
\[(\delz + ik_z)(\delbarz + i k_{\zbar}) + V\]
with the same eigenvalue. Heuristically, twisted eigenfunctions of the untwisted Hamiltonian equal periodic eigenfunctions of a twisted Hamiltonian. With differential forms,  the twisted Hamiltonian is 
\[
    H_L = (\de + k_z \de z + k_{\zbar} \de \zbar )^*(\de + k_z \de z + k_{\zbar} \de \zbar ) + V.
\]
Carrying this to the hyperbolic case, we can once again decompose any Bloch wave as $\psi_k(z) = s(z)u(z)$ where $s(z)$ is a nowhere vanishing ``phase factor" satisfying the proper factors of automorphy, and $u(z)$ is periodic. For a given $\chi_k$, pick one such $s(z)$. Denoting the ``multiplication by s(z)" operator with $s$, Leibniz's rule for a derivation $\delta$ manifests as
\[
    \delta s = s (\delta + \frac{\delta s}{s}). 
\]
If $\psi_k$ is an eigenfunction of the Hamiltonian  $\de^* \de + V = \frac{1}{g}\delz \delbarz +V$ with eigenvalue $\lambda$, then this rule says
\[
    \lambda s u = (\frac{1}{g}\delz \delbarz + V)s u  = s(\frac{1}{g}(\delz+\frac{\delz s}{s})(\delbarz+\frac{\delbarz s}{s}) + V)u = s H_L u.
\]
The one-form $A$ in the theorem is
\[A = \frac{\delz s}{s} \de z + \frac{\delbarz s}{s} \de \zbar = \frac{\de s}{s} = \de \log s.\] 

We must next check that $\de s / s$ obeys all properties listed in the theorem. First, both $\de s$ and $s$ are factors of automorphy for the representation $\chi_k$, so their ratio is a well defined  one-form on $\Sigma$. It is also closed. Its monodromy is best studied on $\HH$, where for a loop $\gamma$ we get 
\begin{align*}
    \int_\gamma \de \log(s) &= \left(\log(s) + 2\pi i \, n  \right)|^{\gamma(p)}_p \\ 
    &= \log\left(\frac{s(\gamma(p))}{s(p)}\right)  + 2\pi i \, n \\
    &= \log\left(\chi_k(\gamma)\right)  + 2\pi i \, n 
\end{align*}
where $n\in\ZZ$ is an arbitrary integer. The exponential removes this ambiguity, resulting in the desired monodromy $\chi_k(\gamma)$

Finally, we must show that every one-form $A$ satisfying the given conditions is equal to $\de s' / s'$ for some nowhere vanishing factor of automorphy $s'$. The form $A-\de s' / s'$ is closed, and its integral around any loop is $2 \pi n$ for some integer $n$.  Therefore, it equals $\de \log p(z) $ for a non-vanishing function $p$. Rearranging, we get 
\begin{equation*}
    A = \de \log(s) + \de \log(p) = \de \log(sp)
\end{equation*}
and so $sp$ is the desired factor of automorphy.
\end{proof}

A physical Hamiltonian must be self-adjoint, which only holds for imaginary $A$. Indeed, the Hermitian conjugate of  $H_L$ is:

\begin{align*}
    H_L^\dagger &= \left(\frac{1}{g} (\delz+A') (\delbarz+A'')\right)^\dagger + V ^\dagger \\
    &= \frac{1}{g}(\delbarz^\dagger+A''^\dagger)(\delz^\dagger+A'^\dagger) + V \\
    &= \frac{1}{g}(-\delz + \bar{A''})(-\delbarz + \bar{A'}) + V,
\end{align*}
where we used $\delz^\dagger = \bar{\delz^*} = -\delbarz$ and likewise for $\delbarz$. Comparing to equation \ref{eq:Local form of abelian hamiltonian}, this equals $H_L$ when $\bar{A''} = -A'$. Without coordinates, this says $\bar{A} = -A$.  This constraint also arises if we only allow phase factors $s(z)$ with magnitude $1$ everywhere. Constraining $s \bar{s} = 1$, we see
\[A' = \frac{\delz s}{s} = - \frac{\delz \bar{s}}{\bar{s}}  \qquad A'' = \frac{\delbarz s}{s} = - \frac{\delbarz \bar{s}}{\bar{s}},\]
again finding $-\bar{A'} = A''$.

\subsubsection{Cohomological interpretation}
This proof suggests that $A = \de s / s$ should only be defined up to a choice of the factor of automorphy $s$. So, the space of $A$ is the space of closed one-forms, modulo exact one-forms and those with trivial monodromy. In terms of cohomology, $A$ lives in $H^1(\Sigma,\CC)/H_1(\Sigma,\ZZ)^*$, and restricting to self-adjoint Hamiltonians limits this to $H^1(\Sigma,\RR)/H_1(\Sigma,\ZZ)^*$.
Poincar\'e duality identifies $H_1(\Sigma,\ZZ)^* \cong H^1(\Sigma, \ZZ)$. The space $H^1(\Sigma,\RR)/H^1(\Sigma,\ZZ)$ carries a natural complex structure isogenic to the Jacobian $H^{0,1}(\Sigma,\CC)/H^1(\Sigma,\ZZ)$. That is, the proper space of possible $A$ is exactly the crystal momentum space.

\subsubsection{Physical interpretation} \label{sec:Hamiltonian physical interpretation}
We could have predicted the form of equation (\ref{eq:momentum_hamiltonian}) from physics. A nonzero crystal momentum implies the electron picks up a phase when traveling along a cycle in $\Sigma$, analogous to a geometric phase (or Aharanov-Bohm phase). In the familiar Aharanov-Bohm effect, a charged particle moving in a loop picks up a phase equal to the exponential of the monodromy of a vector potential. Vector potentials enter the Hamiltonian through minimal coupling, replacing the derivative $d$ with $d + A$, just as in equation (\ref{thm:momentum}). The magnetic field $B = \de A$ is dual to the magnetic flux through the surface, which for us is zero, meaning that $A$ must be closed. $A$ instead acts like magnetic fluxes threading through the holes of $\Sigma$, as indicated in figure $\ref{fig:magnetic_flux}$. The flux quantity determines the monodromy and thus the crystal momentum.

Which vector potential should we use? The Hamiltonian should be invariant under gauge transforms, but the only gauge-invariant quantities in electromagnetism are the magnetic field and the monodromy. These are both fixed, so all closed one-forms with the prescribed monodromy should be equivalent. This agrees with the analysis in the last subsection. The gauge freedom is equal to the freedom in the choice of phase factor $s$ (with magnitude $1$ everywhere).
\begin{figure}[htp]
    \centering
    \includegraphics[width = \textwidth]{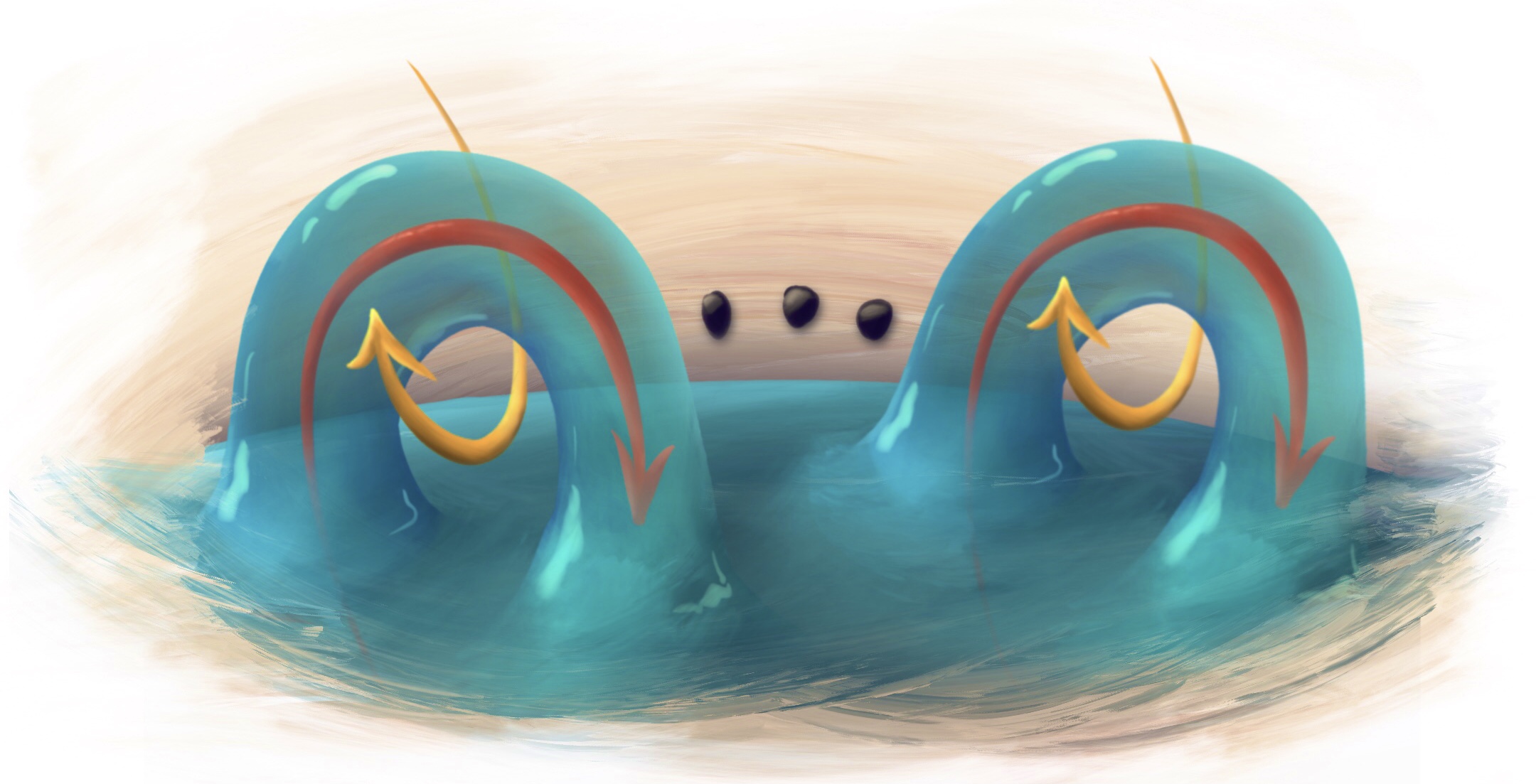}
    \caption{The crystal Hamiltonian can be expressed geometrically though closed one-forms with periods equal to the given crystal momentum. Physically, these are vector potentials with zero magnetic flux through the Riemann surface. This corresponds to magnetic fluxes threading the holes of the Riemann surface, as indicated. Said differently, these arrows designate a basis of harmonic one-forms. }
    \label{fig:magnetic_flux}
\end{figure}

\subsubsection{Geometric interpretation}\label{sec:geometric interpretation}
The connection with electromagnetism suggests a gauge theoretic approach, returning us to the realm of geometry. Electromagnetism is a $U(1)$ gauge field, and the vector potential is a $U(1)$ principle connection whose curvature is the magnetic field. The magnetic field here is zero, so the connection is flat. This is exactly the setup from Section \ref{sec:Riemann Surfaces}.  The differential operator $\de + A = \nabla_L$ is the flat connection, so the Hamiltonian gains the geometric interpretation 
\begin{equation} \label{eq:abelian hamiltonian geometric form}
    H_L = \nabla_L^* \nabla_L + V.
\end{equation}
Here, the adjoint is taken with respect to the standard Hermitian metric pulled back from $\HH \times \CC$. The Hamiltonian acts on sections of $L$, so its Hilbert space is $L^2(L)$. The eigenfunctions of $H$ in $\calH_k$ are the eigensections of $H_L$. Trivializing the line bundle gives an isomorphism $L^2(L) \cong L^2(\Sigma)$, which swaps the Hamiltonians of equation \ref{eq:abelian hamiltonian geometric form} and \ref{eq:momentum_hamiltonian}.

\subsection{Crystal Hamiltonian for nonabelian Bloch states} \label{sec:Nonabelian hamiltonian}
We can find similar expressions for the Hamiltonian in the nonabelian case. Consider the representation $\rho: \pi_1(\Sigma) \to U(n)$, associated to the vector bundle $E \to \Sigma$. The Hamiltonian restricted to the subspace of nonabelian Bloch states that transform according to $\rho$ is denoted $H_\rho$.

\begin{theorem}\label{thm:nonabelian hamiltonain}
Let $A$ be an $\End(E)$-valued  one-form with monodromy $\rho(\gamma)$ about every loop $\gamma$. Then, every eigenfunction of $H_\rho$ uniquely corresponds to an eigensection with equal eigenvalue of $H_E = (\de + A)^* (\de + A) + V \cdot \textrm{Id}$ acting on sections of $E$.
\end{theorem}

\begin{proof}
Following the proof of theorem \ref{thm:momentum}, we start by separating out a phase factor and a periodic part. We accomplish this by picking a frame. Let $\tilde{s}_1(z), \dots, \tilde{s}_n(z)$ be a set of complex functions with magnitude 1 on $\HH$ transforming according to $\rho$. 
This pushes down to a frame $s_1,\dots,s_n$ of $E$, unitary with respect to the induced Hermitian metric $h_\flat$ (this frame always exists because $E$ is topologically trivial). Every section $\vec{\psi}$ of $E$ comes from a function $\tilde{\psi}(z)$ on $\HH$, written as 
\[\tilde{\psi}(z) = \tilde{s}_1(z) \tilde{u}_1(z) + \dots + \tilde{s}_n(z) \tilde{u}_n(z)\]
for $\tilde{u}_i(z)$ periodic functions on $\HH$. These are lifts of functions $u_i$ on $\Sigma$. The exterior derivative acts on $\tilde{\psi}$ by

\begin{gather}
    \tilde{\psi} = \tilde{s}_1 \tilde{u}_1 + \dots + \tilde{s}_n \tilde{u}_n = \begin{pmatrix} \tilde{s}_1 & \dots & \tilde{s}_n \end{pmatrix} \begin{pmatrix} \tilde{u}_1 \\ \vdots \\ \tilde{u}_n \end{pmatrix} \nonumber \\
    \de \tilde{\psi} = \tilde{s}_1 (\de + \frac{\de \tilde{s}_1 }{\tilde{s}_1 })\tilde{u}_1 + \dots + \tilde{s}_n (\de + \frac{\de \tilde{s}_n }{\tilde{s}_n })\tilde{u}_n  = 
    \begin{pmatrix} \tilde{s}_1 & \dots & \tilde{s}_n \end{pmatrix}
    \begin{pmatrix} 
    \de + \frac{\de \tilde{s}_1 }{\tilde{s}_1 } &  &\\
    & \ddots & \\
    & & \de + \frac{\de \tilde{s}_n }{\tilde{s}_n } \end{pmatrix}
    \begin{pmatrix} \tilde{u}_1 \\ \vdots \\ \tilde{u}_n \end{pmatrix}.  \label{eq:derivative hyperbolic nonabelian}
\end{gather}

Said another way, we can consider $\tilde{\psi}$ to be a section of the trivial vector bundle $\HH \times \CC^n$. This carries the standard constant frame of basis vectors $e_i$ of $\CC^n$, and we have picked a unitary frame $s_i = s_i(z) e_i$. This has a diagonal change of frame matrix 
\[S = \begin{pmatrix} 
    s_1(z) &  &\\
    & \ddots & \\
    & & s_n(z) \end{pmatrix}.\]
The computation in equation \ref{eq:derivative hyperbolic nonabelian} notes that, in the frame $\tilde{s}$, the trivial connection takes the form $\de + S\inv \de S$. 
Pushing down to $E$, we get a global frame $s$, which we can write with respect to a local constant frame. The change of frame matrix is the same, as is the coordinate form of the trivial connection in this frame, $\de + S\inv \de S$. We see the Hamiltonian acts on a function $\tilde{\psi}$ as 
\[H \tilde{\psi} = (\de^* \de + V) \vec{\psi} = \begin{pmatrix} \tilde{s}_1 & \dots & \tilde{s}_n \end{pmatrix} \left( (\de + S^{-1} \de S)^*  (\de + S^{-1} \de S) + V      \right) \vec{u}.\]
Pushing forward to the associated section  $\vec{\psi}$ of $E$, this has much the same form in the frame $s$:
\[H_E \vec{\psi} = \left( (\de + S^{-1} \de S)^*  (\de + S^{-1} \de S) + V \cdot Id     \right)\vec{\psi}.\]

Gluing together the local constant frames of $E$, $S\inv \de S$ glues into a global $\End(E)$-valued one-form on $\Sigma$. Using the reasoning from \ref{thm:momentum}, the monodromy of $\frac{\de s_i}{s_i}$ is $\rho(e_i)$, so the full monodromy of $S\inv \de S$ is $\rho$, as desired.

Finally, we must contend with gauge transformations, which come from a different choice of global frame $s'$. We call the change of frame $s \to s'$ matrix $D$, which is a global section of $\End(E)$. This gauge transform replaces $\de + S\inv \de S$ with $\de + D\inv S\inv \de S D + D\inv \de D$. However, $D$ is a global endomorphism so the monodromy of $D\inv \de D$ must be trivial. Since the only gauge invariant quantity is monodromy, every $A$ with given monodromy (up to conjugation) is representable as $D\inv S\inv \de S D + D\inv \de D$ for some endomorphism $D$.
\end{proof}

\begin{remark}
For abelian crystal momenta, this Hamiltonian also arose from minimal coupling with a $U(1)$ gauge field, see Section \ref{sec:Hamiltonian physical interpretation}. In higher rank, it should couple to a $U(n)$ gauge theory. The flat connections are solutions to the Yang-Mills equations on $\Sigma$, which we can describe with $U(n)$ fluxes threading the holes of $\Sigma$ as in Figure \ref{fig:magnetic_flux}. 
\end{remark}

\subsubsection{Geometric interpretation}

Much like the abelian case, we can think of $(\de + A)$ as a flat unitary connection $\nabla_E$ on $E$. The Hamiltonian has geometric form $H_E = \nabla_E^* \nabla_E + V$. This lets us state the central problem of hyperbolic band theory in full generality:

\begin{problem}[The band theory problem]
Consider a Riemann surface $\Sigma$ with a degree zero, Hermitian vector bundle $(E,h)$ with $U(N)$ flat connection $\nabla_E$ and a real potential $V:\Sigma \to \RR$. How does the spectrum of the operator
\begin{align*}
H_E = \nabla_E^* \nabla_E +V
\end{align*}
vary along the moduli space of flat irreducible unitary connections?
\end{problem}

We can bring this closer to algebraic geometry by trading flat structures with holomorphic structures. The flat connection splits into holomorphic and antiholomorphic components, $\nabla_E = \nabla_E' + \nabla_E''$. The $(0,1)$ part defines a Dolbeault operator $\delbar_E = \nabla_E''$, making $E$ into a holomorphic vector bundle.  This map is a diffeomorphism between the moduli space of flat irreducible unitary connections and the moduli space of stable holomorphic vector bundles. We can also express the Laplacian through $\delbar_E$. The Laplacian decomposes as 
\[\nabla_E^*\nabla_E = (\nabla_E' + \nabla_E'')^*(\nabla_E' + \nabla_E'') = \nabla_E'^*\nabla_E' +\nabla_E''^*\nabla_E''.\]
The cross terms vanish due to type considerations. A Weitzenböck-type identity states
\footnote{Specifically, this is the Nakano-Akizuki formula. The usual expression on the right is $[F_{\nabla_E},\Lambda]$, which simply equals $\star F_{\nabla_E}$ on a Riemann surface.} 
\begin{align} \label{eq:Nakano-Akizuki}
    \nabla_E''^*\nabla_E'' = \nabla_E'^*\nabla_E' + i \star F_{\nabla_E}.
\end{align}
Since $\nabla_E$ is flat, the curvature $F_{\nabla_E}$ vanishes and $\nabla_E^*\nabla_E = 2\nabla_E''^*\nabla_E'' = 2 \delbar_E^*\delbar_E $. This gives an alternate version of the hyperbolic band theory problem:

\begin{problem}[The band theory problem, holomorphic version]
Consider a Riemann surface $\Sigma$ with a degree zero holomorphic vector bundle $(E,\delbar_E)$ and a real potential $V:\Sigma \to \RR$. How does the spectrum of the operator
\begin{align*}
H_E = 2\delbar_E^* \delbar_E +V
\end{align*} vary along the moduli space of stable vector bundles $\calN^s(\Sigma,n)$?
\end{problem}

The space of hyperbolic crystal momenta (the hyperbolic Brillouin zone) is the moduli space of all stable bundles. The connected components are stable bundles of a specified rank, starting with the abelian Brillouin zone (Jacobian of the curve), but increasing in dimension and complexity as the rank increases (See Figure \ref{fig:stable}). We wish to find the band structure: A graph of the spectrum of $H_E$ expressed as a many-sheeted cover of hyperbolic Brillouin zone.

We note one more geometric perspective on this Hamiltonian. As the Hermitian metric induced by the standard one on $\HH\times \CC^n$ is flat, its Chern connection relative to the holomorphic structure $\delbar_E$ equals the flat connection $\nabla_E$. Using equation \ref{eq:Nakano-Akizuki}, we can control a scalar part of the Laplacian with the Hermitian metric. Applied to the Chern connection $\nabla_h$, it says
\[\nabla_h^* \nabla_h = \delbar_E^* \delbar_E + i \star \Theta(h),\]
where $\Theta(h)$ is the curvature of the metric. Choosing $h$ with a central curvature scaled by the potential, we have that $i \star \Theta(h) = \frac{1}{2}V \cdot \textrm{Id}$. This induces a third formulation of the band theory problem:

\begin{problem}[The band theory problem, Chern version]
Consider a Riemann surface $\Sigma$ with a degree zero topological vector bundle $E$ and fixed Hermitian metric $h$ with curvature $\frac{1}{2}\star V$, and denote the Chern connection by $\nabla_h$. How does the spectrum of the operator
\begin{align*}
H_E = \nabla_h^*\nabla_h
\end{align*} vary with the holomorphic structure of $E$?
\footnote{Note that the adjoint is taken with respect to the curved metric $h$. Since $\Theta(h)$ is central, we can find a unitary frame with respect to the metric $h_\flat$ induced by $\HH \times \CC^n$, such that $h$ is a multiple of the identity matrix. As such, the adjoint with respect to $h$ agrees with the adjoint with respect to $h_\flat$.}
\end{problem}

This formulation is natural in the context of the Higgs bundle approach described in the remainder of this paper, and is perhaps the best framework for more general questions. That said, the potential will suffice for our purposes and will simplify some exposition.

\begin{figure}[htp]
    \centering
    \includegraphics[width = \textwidth]{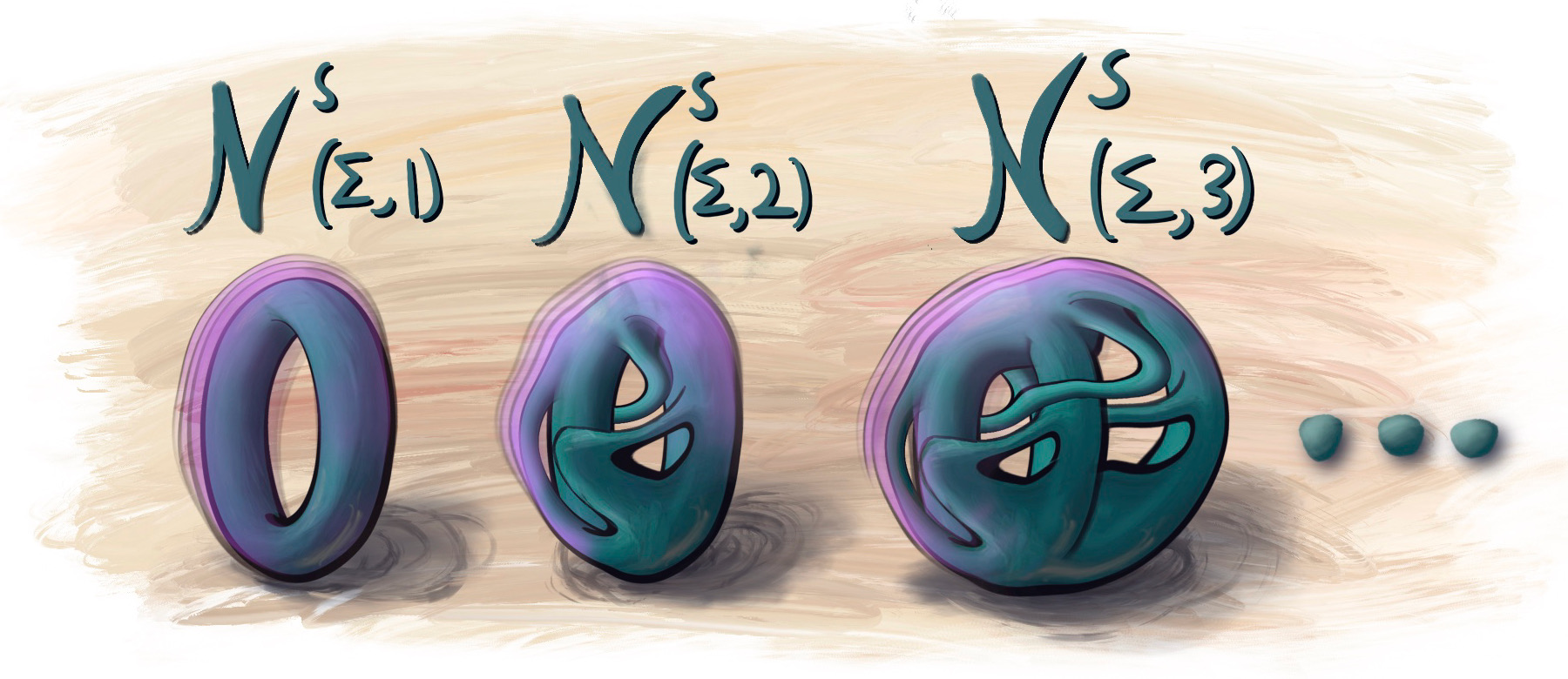}
    \caption{Visualization of the Brillouin zone and band structure of a hyperbolic crystal. The Brillouin zone is blue, with connected components consisting of the moduli space of stable vector bundles over $\Sigma$ with degree zero and rank $n$, $\calN^2(\Sigma,n)$. These increase in dimension and complexity with $n$, starting with a torus (the Jacobian) at $n=1$. (The drawings of higher rank moduli spaces only seeks to communicate the increasing complexity). The band structure is a locally an infinite-sheeted cover of the Brillouin zone, the first few sheets visualized here with purple shells.}
    \label{fig:stable}
\end{figure}

\section{Higgs bundles}
\label{sec:Higgs_Bundles}
The primary intent of this paper is to study the band theory problem from the last section using Higgs bundles. In this section, we introduce the core ideas of Higgs bundles, taking a route optimized for hyperbolic band theory.

\subsection{Riemann surfaces as spectral curves}\label{spectral}

For motivation, recall that a Riemann surface $\Sigma$ encodes a hyperbolic crystal with abelian crystal momentum. Suppose the crystal had a symmetry, represented by a finite group $G$ acting holomorphically on $\Sigma$. To remove redundancy, we represent the crystal data as a periodic geometric structure over a fundamental domain of $G$. The fundamental domain is the quotient $C = \Sigma/\Gamma$, which inherits the Riemann surface structure of $\Sigma$. It comes with a $|G|$-to-1 branched covering map $p:\Sigma \to C$. We want to encode the branched covering $p$ and line bundle $L \to \Sigma$ in a geometric structure over $C$. Over a regular point $b\in C$, this data consist of a set of points $p\inv (b)$ that we represent as 'heights' $\lambda_i \in \CC$, each with a one dimensional subspace $L|_{p_i}$. We package this into a $N\times N$ matrix acting on the vector space $\oplus_{i} L|_{p_i}$, with eigenvalues $\lambda_i$ and eigenvectors $L|_{p_i}$. 

Roughly speaking, A Higgs bundle is a \textit{global, holomorphic} version of this construction. It consists of a holomorphic vector bundle $E \to C$, and a Higgs field: a holomorphic section $\phi$ of $\End{E}\otimes K$, where $K = T^{*(1,0)}(C)$ is the holomorphic cotangent bundle of $C$, also called the canonical bundle. After choosing a frame for $E$, $\phi$ is a matrix of holomorphic one-forms, with eigenvalues valued in $K$. These generalize to twisted Higgs bundles, whose eigenvalues can live in any line bundle $K(D) = K \otimes \calO(D)$, where $\calO(D)$ is the line bundle associated to some divisor $D$ on $C$. The graph of the eigenvalues defines a \textit{spectral curve} $\Sigma$ as a codimension one submanifold of the total space $\Tot(K(D))$

\begin{figure}[htp]
    \centering
    \includegraphics[width = \textwidth]{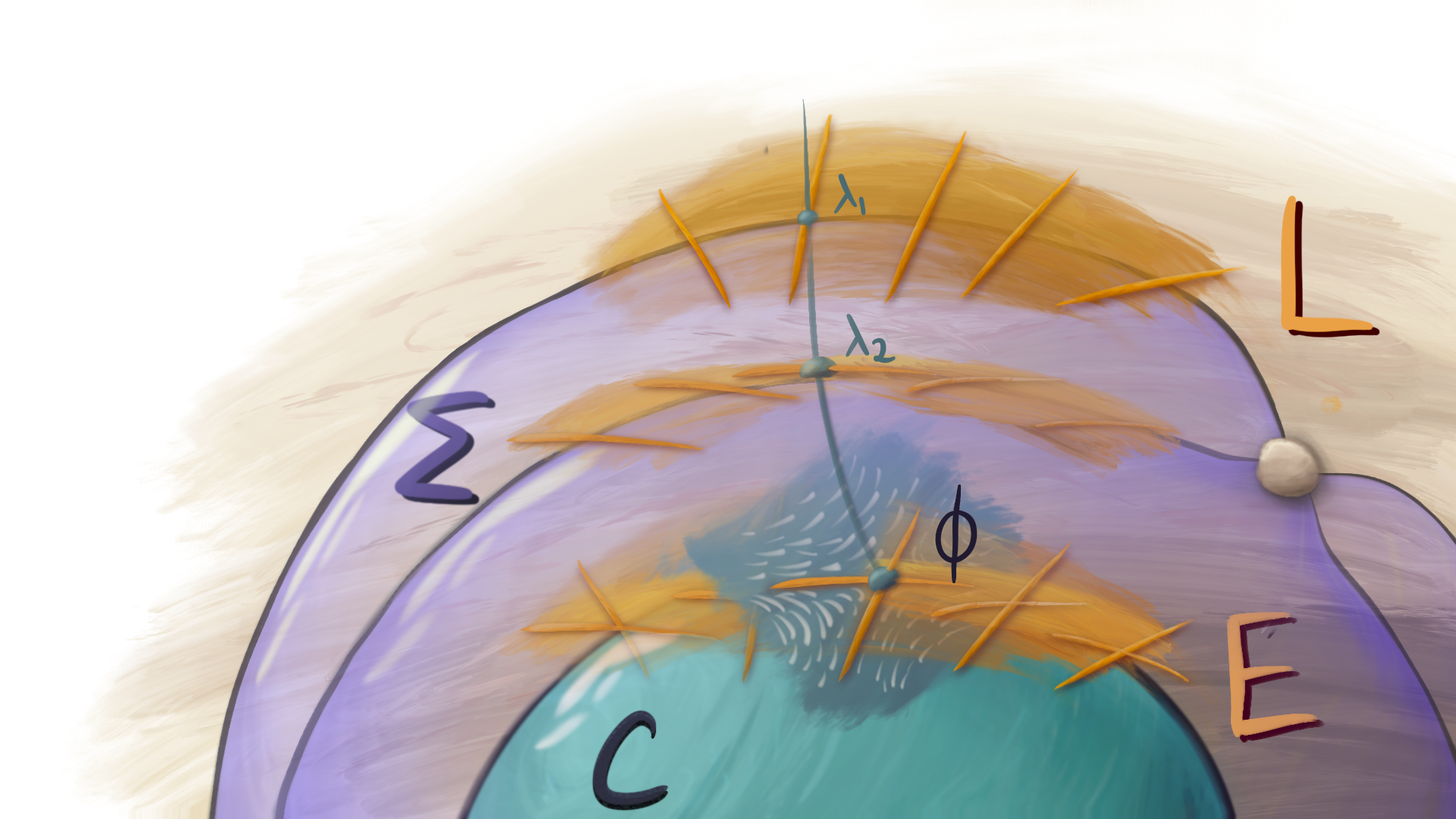}
    \caption{The spectral data of a Higgs bundle $\phi$ encodes a line bundle $L$ on a branched cover $\Sigma \to C$. Locally $\phi$ is a matrix-valued one-form. Its eigenvalues encode the layers of $\Sigma$ over $C$, and its eigenvector decomposition splits the vector bundle $E \to C$ into a line bundle on $\Sigma$. Generically, repeated eigenvalues only occur in the same Jordan block of the Higgs field, so have a single eigenvector. These means the line bundle smoothly extends over the branch points (white dot).}
    \label{fig:spectral}
\end{figure}

We can pass between a twisted Higgs bundle $(E,\phi: E \to E \otimes K(D))$ and its spectral data $(\Sigma \sub \Tot(K(D)), L \to \Sigma)$, a process called \textit{abelianization}. Starting from the spectral data, $E$ is related to $L$ at a regular point $b$ by $E|_b = \oplus_{i} L|_{p_i}$. We extend this across branch points using the language of sheaves. Denoting the sheaf of holomorphic sections of $L$ by $\calO(L)$, we define the pushforward sheaf\footnote{This is often called the ``direct image sheaf"}
$p_*\calO(L)$: For $U\sub C$ open, the space of holomorphic sections $H^0(U,p_*\calO(L))$ equals the holomorphic sections of $L$ on $p\inv(U)$, $H^0(p\inv(U),\calO(L))$. For a branched cover $p$ the pushforward sheaf is locally free, defining a pushforward vector bundle $E = p_*L$. The pushforward uniquely maps holomorphic sections of $L$ to those of $p_*L$, allowing one to define pushforwards of operators. For a good introduction, see \cite[chapter 2]{hitchin_integrable_2013}. We can also obtain the Higgs field. The total space of $K(D)$ has projection $\pi: \Tot(K(D)) \to C$, and carries the tautological line bundle $\pi^*K(D)$. It has a tautological section $\lambda$, which fiberwise is the identity function $\lambda(z) = z$. Multiplying by $\lambda$ acts fiberwise on $L$, scaling $L|_{p_i}$ by $\lambda_i$. The pushforward of this operation gives the Higgs field $\phi : p_*L \to p_*L \otimes K(D)$. 

From a Higgs bundle, we get a spectral curve $\Sigma$ by graphing the eigenvalues of $\phi$ in $\Tot(K(D))$. More invariantly, $\Sigma$ is the locus of zeros of the characteristic polynomial $\det(\phi-\lambda I) = \lambda^d + a_{1}\lambda^{d-1} + \dots + a_d $, with coefficients $a_i = \Tr(\phi^i) \in H^0(\Sigma, K(D)^i)$. These coefficients characterize the spectral curve, giving a fibration on the moduli space of Higgs bundles (see Section \ref{sec:Moduli of Higgs bundles}). The projection $\Tot(K(D)) \to C$ restricts to the branched covering $p: \Sigma \to C$. Finally, the line bundle on $\Sigma$ is the unique $L$ with pushforward $p_*L=E$.

\subsection{Hyperelliptic curves} \label{sec:Hyperelliptic}
Let us see this apply this to hyperelliptic curves, double branched covers of the Riemann sphere $\PP^1$. Hyperelliptic curves are an important simple case, which we will use as an example throughout this paper.

Consider a crystal with spatial inversion symmetry. Inversion about a point $z=0$ is an isometry $\sigma$ locally acting as $z \to -z$. In the Poincar\'e disk model centered about this point, inversion is Euclidean inversion of the disk. A crystal with spatial inversion symmetry has a fundamental polygon symmetric about the origin, as depicted in figure  \ref{fig:hyperelliptic}A. The fundamental domain for inversion is half the polygon. The quotient $\Sigma / \sigma$ is constructed by cutting out the fundamental domain and gluing the boundaries according to $\sigma$. Each side of the half polygon is identified with its flipped self up to lattice translation, as shown in Figure \ref{fig:hyperelliptic}B. Folding this up gives a topological sphere, shown in Figure \ref{fig:hyperelliptic}C. More formally, the branch points of $p:\Sigma \to \Sigma/\sigma$ are the fixed points of $\sigma$, consisting of the center, vertex, and midpoint of half the edges of the polygon. If $g$ is the genus of $\Sigma$, the polygon has $4g$ edges, and $\sigma$ has $2g+2$ fixed points. The Riemann-Hurwitz formula guarantees that a double branched cover $p:\Sigma \to C$ with $2g+2$ branch points must have $C = \PP^1$. So, the Riemann surface is a double-cover of $\PP^1$, depicted by Figure \ref{fig:hyperelliptic}D. These are called hyperelliptic curves, and $\sigma$ is a hyperelliptic involution. Conversely, any hyperelliptic curve can be represented with a spatially symmetric unit cell \cite{gallo_uniformization_1979}.

\begin{figure}[htp]
    \centering
    \includegraphics[width = 1.\textwidth]{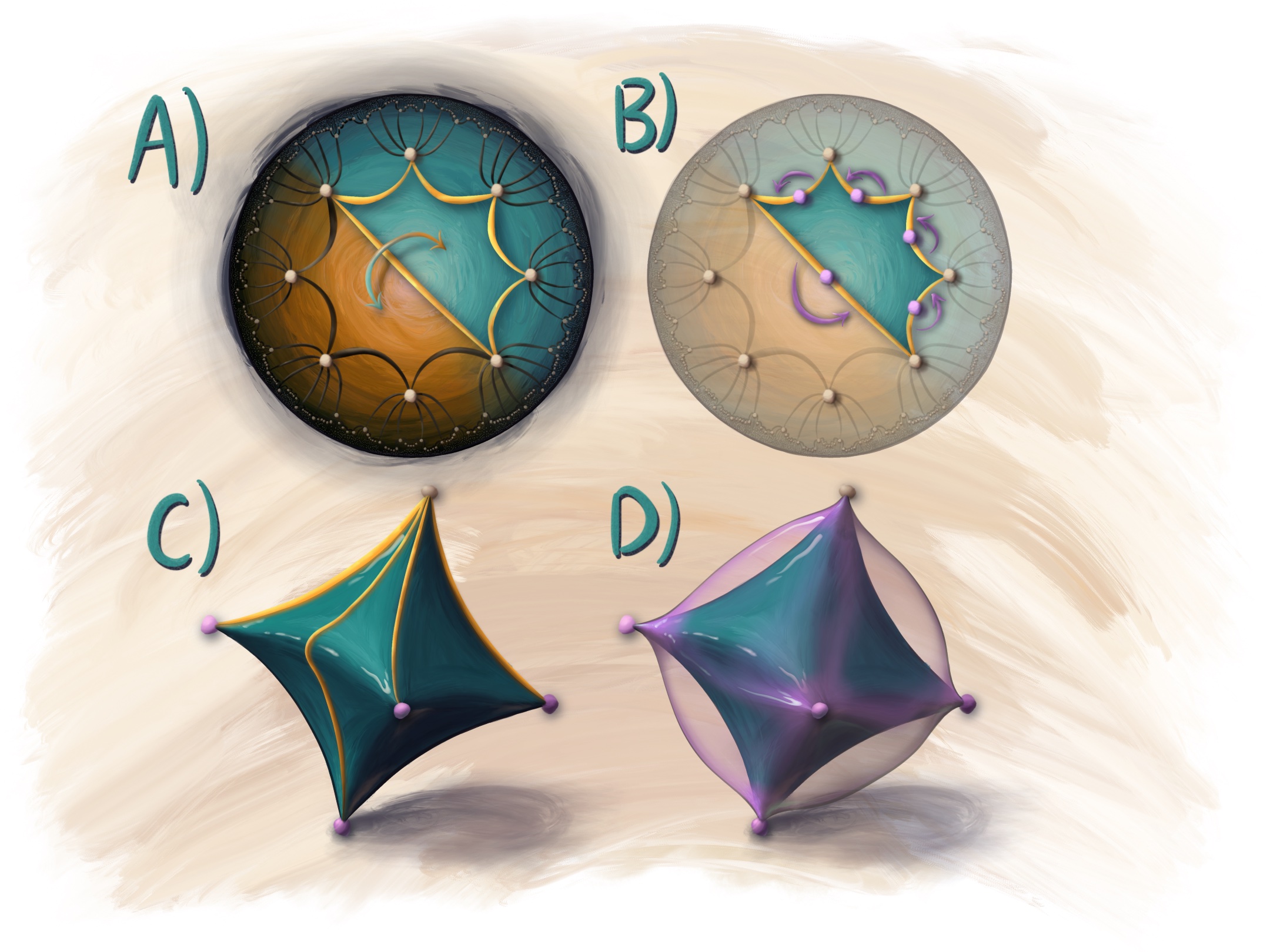}
    \captionsetup{singlelinecheck=off}
    \caption[test]{Crystals with spatial inversion symmetry and hyperelliptic curves.
    \begin{enumerate}[label=(\Alph*)]
    \item Spatial inversion interchanges both halves of the unit cell, which are minimal generating sets of the symmetry. 
    \item Restricting to half the cell, spatial inversion acts by ``folding'' each edge in half. The fixed points are marked by purple and white dots. 
    \item The half-cell topologically folds into a sphere, with the final location of fixed points and edges indicated in the figure. It inherits a constant negative curvature away from the fixed points, and cone points with angle $\pi$ at the fixed points.
    \item All symmetric crystals holomorphically double cover $\PP^1$, corresponding to hyperelliptic curves. This figure depicts $\Sigma$ as a double cover of $\PP^1$.
    \end{enumerate} }
    \label{fig:hyperelliptic}
\end{figure}

We can represent the crystal data $(\Sigma,L)$ as a twisted Higgs bundle $(p_*L,\phi)$ over $\PP^1$, with $\phi$ valued in $K(D)$. The characteristic polynomial of $\phi$ is $\lambda^2 + P(z)$, where $P(z)$ is a section of $K(D)^2$ with zeros at the branch points. The degree of the polynomial $P(z)$ is $2g+2$, so the degree of $K(D)$ is $g+1$. On $\PP^1$, this uniquely characterizes $K(D) = \calO(g+1)$.  Relating the Euler characteristics of the bundles $L$ and $p_*L$ shows that the degree of $p_*L$ is $-(g+1)$. Since $p_*L$ lives over $\PP^1$, the Birkhoff-Grothendieck theorem says it splits as
\begin{equation*}
    p_*L = \calO(-k)\oplus \calO(k-(g+1))
\end{equation*}
for some integer $k$. Using this decomposition, we can write $\phi$ as a matrix of sections of different line bundles over $\PP^1$:

\begin{gather}
    \phi \in H^0\left(\Sigma,
    \begin{pmatrix}
        \calO(-k) \otimes \calO(-k)^* & \calO(-k) \otimes \calO(k-(g+1))^* \\
        \calO(-k)^* \otimes \calO(k-(g+1)) & \calO(k-(g+1)) \otimes \calO(k-(g+1))^*
    \end{pmatrix}
    \otimes \calO(g+1) \right) \nonumber \\ 
    =H^0\left( \Sigma, 
    \begin{pmatrix}
        \calO(g+1) & \calO(2(-k+g+1)) \\
        \calO(2k) & \calO(g+1)
    \end{pmatrix} \right). \label{eq: Hyperelliptic Higgs}
\end{gather}
A $2\times2$ matrix of polynomials with these degrees represents a Higgs field. A (stable) Higgs field only exists if all of the entries have nonnegative degree, implying $0\leq k \leq (g+1)$.

\subsubsection{Parabolic Higgs bundles}
\label{sec:parabolic}
Parabolic Higgs bundles offer a dual perspective to twisted Higgs bundles. Instead of a holomorphic section of $\End(E)\otimes K(D)$, $\phi$ is a meromorphic section of $\End(E)\otimes K$, with poles allowed on the divisor $D$. The Higgs field only blows up along a given direction in the fiber of $E$. A ``parabolic bundle" encodes how the Higgs field diverges:

\begin{definition}[Rank 2 parabolic bundle]
Consider a \textit{parabolic divisor} $D = \{d_1,\dots, d_n\}$ on $\Sigma$. For a rank $2$ vector bundle $E \to \Sigma$, a parabolic structure on $E$ assigns to each parabolic point $d_i$ a $1$-dimensional subspace $F_i \sub E|_{d_i}$ and weights $\Vec{\alpha} = (\alpha_1,\alpha_2)$ with $0 \leq \alpha_1 < \alpha_2 < 1 $. 
\end{definition}

\begin{remark}
Higher rank parabolic bundles are characterized by a full flag on the fibers at the parabolic divisor, with weights. To simplify the exposition, we restrict to rank $2$ parabolic bundles throughout the paper.
\end{remark}

A parabolic Higgs bundle features a Higgs field respects this structure. Namely, the residue of such a parabolic Higgs field $\phi$ at a point $d_i\in D$ sends $E|_{d_i}$ to $F_i$ and $F_i$ to $0$ --- thus, it is nilpotent at $D$. First, this implies $\Tr(\phi)$ is a holomorphic section of $K(D)$ vanishing along $D$ by nilpotency, so it descends to a holomorphic section of $K$. For a hyperelliptic curve where $C=\PP^1$,  the canonical bundle is $\calO(-2)$, so the only global holomorphic section is identically zero. Thus, $\phi$ is trace-free, and so the structure of the Higgs bundle is given by the group $SL(2,\CC)$. Secondly, $P(z) = \det(\phi^2)$ vanishes on $D$, and so $D$ is contained in the set of branch points. The degree of $K(D)$ is $g+1$, from which it follows that $D$ contains $g+3$ points that are amongst the $2g+2$ branch points. The other branch points occur where $\phi$ itself is nilpotent.

\subsubsection{Orbifolds}
Parabolic bundles in hyperbolic band theory come from orbifolds. These arise from a manifold with  properly discontinuous (but not necessarily free) finite group action. They are like manifolds, except they are locally modeled on $\RR^n/G$ for a finite group action $G$. For example, take a hyperelliptic curve $\Sigma$ with hyperelliptic involution $\sigma$. The quotient $\Sigma/\sigma$ is an orbifold, modeled on $\CC$ away from the branch points and on $\CC/\{z \to -z\}$ about the branch points. However,  the holomorphic map $z \to z^2$ sends $\CC/\{z \to -z\}$ to $\CC$, so these charts define the same Riemann surface structure, giving a double cover of $\PP^1$.

We cannot ignore the orbifold structure, as the metric is not isomorphic after quotienting by $z \to -z$. It instead develops a conical singularity with cone angle $\pi$. For a papercraft analogy, $\CC/\{z \to -z\}$ is formed by cutting the plane in half and ``folding" the free side about zero, forming a cone. This gives the local structure of $\Sigma/\sigma$ at a branch point. All in all, the resulting $\PP^1$ inherits a metric with constant negative curvature and cones at all the branch points, as depicted in Figure \ref{fig:hyperelliptic}C. The vector bundle $E\to \Sigma/\sigma$ carries a natural $\ZZ_2$ action of $\sigma$, which induces a parabolic structure on $E$. The parabolic divisor consists of fixed points of $\sigma$ (orbifold points), and the distinguished lines are those preserved by $\sigma$. Thinking of $E$ as the pushforward $p_*L$, the preserved line is the pushforward of even sections of $L$. From the singularity in the metric $g(z)~|z|$, the parabolic weights are $(0,1/2)$. This is a particularly simple example of a very general correspondence between orbifold bundles and parabolic bundles, see \cite{biswas_parabolic_1997}. Note that for an $L$ of degree $0$, $p_*L$ has degree $-d/2$, where $d = 2g+2$ is the degree of the parabolic divisor of branch points. With weights $(0,1/2)$, the parabolic degree of $p_*L$ is $-\deg(p_*L)+\sum \alpha_i= -d/2+d/2=0$. 

This parabolic structure is reflected in the local form of a Higgs field. Choosing a coordinate $w$ on $\PP^1$ with $w=0$ a branch point, $\Sigma$ has a local coordinate $z$ related by $z^2 = w$. We can write a local holomorphic section $f(z)$ in these coordinates by splitting into even and odd parts, $f(z) = f^e(z^2)+z f^o(z^2)$. $f^e$ and $f^o$ provide a local frame for the pushforward bundle $E$. The Higgs field multiplies by the eigenvalue $\pm \sqrt{P(w)} ~ z$, so the matrix of the Higgs field as a $K(D)$ valued Higgs bundle is 
\[\phi =\begin{pmatrix}
 0 & w \\ 1 & 0
\end{pmatrix}.\]
We retrieve $\phi$ as a $K$-valued Higgs field by dividing by a section of $\calO(D)$, namely the polynomial $P(w)$. Since no roots are repeated, this is locally proportional to $w$, and so the Higgs field is 
\[\phi =\begin{pmatrix}
 0 & 1 \\ \frac{1}{w} & 0
\end{pmatrix} \de w.\]
The residue of $\phi$ is a nilpotent matrix, with the distinguished line consisting of even sections. So, the Higgs field respects the orbifold parabolic structure.

\subsection{Moduli of Higgs bundles and nonabelian Hodge theory} \label{sec:Moduli of Higgs bundles}

\begin{figure}[htp]
    \centering
    \includegraphics[width = \textwidth]{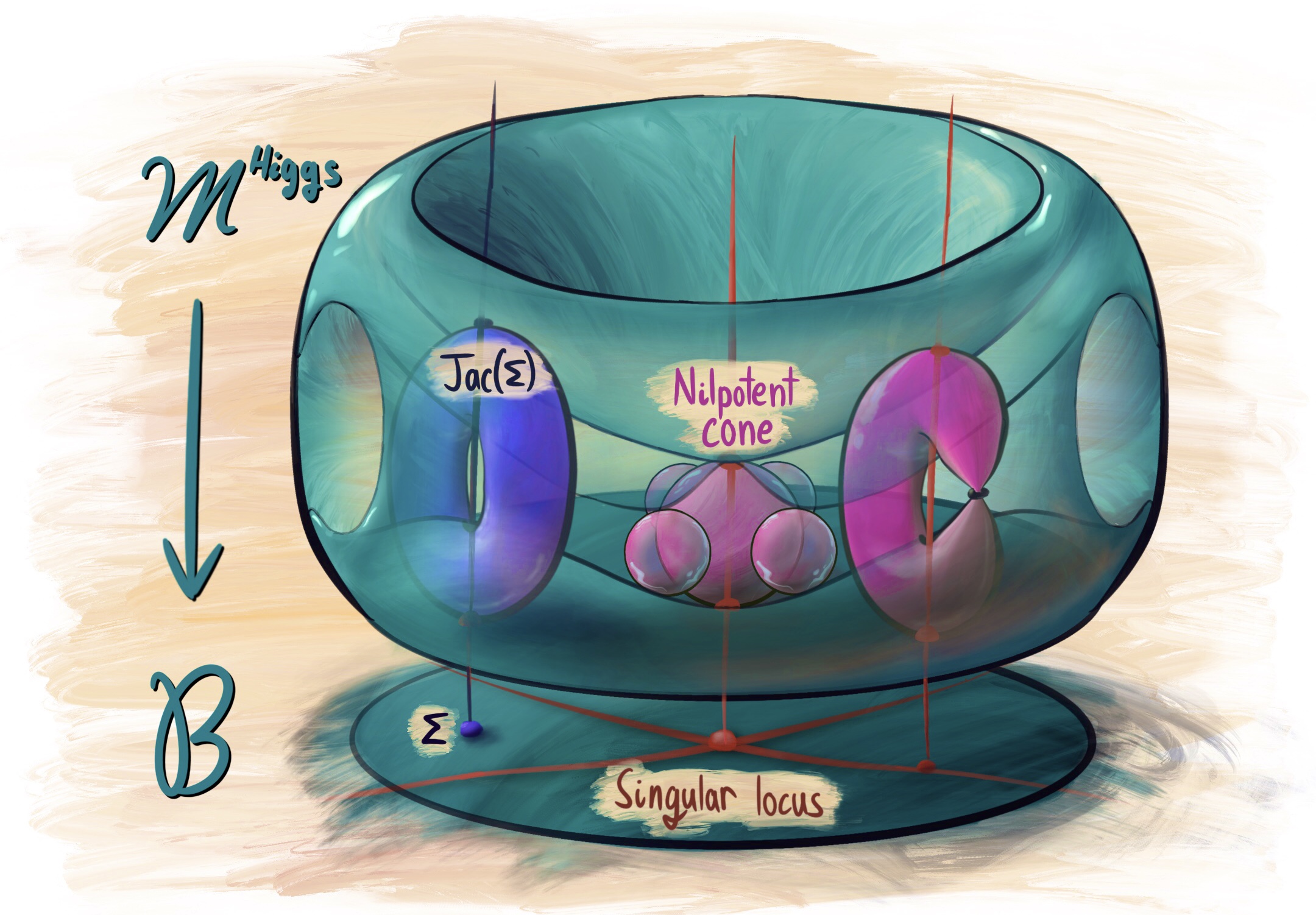}
    \caption{Illustration of the moduli space of Higgs bundles $\MHiggs$. The Hitchin fibration maps $\MHiggs$ to the Hitchin base $\calB$, a vector space of half the dimension parametrizing different spectral curves. A generic point in $\calB$ gives a smooth spectral curve, with fiber the Jacobian (shown in purple). On the singular locus (shown in red), the spectral curve is singular, and the fibers can degenerate (indicated by the pink pinched torus). The most degeneration occurs over $0\in B$, where the fiber is called the nilpotent cone (indicated here by the $5$ pink spheres). This drawing is informed by the moduli space of Higgs bundles with 4 parabolic points over $\PP^1$, which Section \ref{sec:genus 1} describes in detail.}
    \label{fig:moduli}
\end{figure}
Sending a Higgs bundle to its spectral curve $\Sigma$ lends the moduli space a fibration structure. The spectral curve is characterized by the coefficients of its characteristic polynomial
\[ (\Tr(\phi),\Tr(\wedge^2\phi),\dots ,\Tr(\wedge^n\phi) ) \in \bigoplus_{i=0}^{n} H^0(C,K(D)^i ) = \calB.\]
This yields the ``Hitchin map'' from the space of Higgs bundles $\MHiggs$ to an affine space $\calB$ called the Hitchin base. A point $b \in \calB$ determines a spectral curve $\sigma_b$. The regular locus  $\calB_\reg$ consists of points $b$ with a smooth spectral curve $\sigma_b$, where the characteristic polynomial has repeated roots. The remaining spectral data is a line bundle $L\to \Sigma_b$. For a smooth $\Sigma_b$, the fiber of the Hitchin map is the space of line bundles $\Jac(\Sigma_b)$.

The regular points $\calB_\reg$ are (Zariski) open and dense in $\calB$, and so we obtain a picture of the Higgs moduli space: equidimensional tori fibered over an affine base, which degenerates along a singular locus $\calB \backslash \calB_\reg$. The most degenerated fiber comes from the most degenerate spectral curve, where the characteristic polynomial is $P(\phi-\lambda)=\lambda^n$. This occurs at the zero in the Hitchin base, and the fiber is called the nilpotent cone. This picture is sketched in Figure \ref{fig:moduli}.

Moreover, $\MHiggs$ inherits a hyperk\"{a}hler structure from the infinite-dimensional space of pairs $(E,\phi)$. The symplectic structure on $\MHiggs$ from the fibration restricts to zero on the fibers, making these isotropic submanifolds. Remarkably, the dimension of the Hitchin base is exactly half that of $\MHiggs$, and so the fibers are Lagrangian tori \cite{hitchin_stable_1987}. The coordinate functions on $\calB$ Poisson commute with one another, making $\MHiggs$ into an algebraically completely integrable system.

Similar properties hold for twisted and parabolic Higgs bundles. The moduli space of twisted bundles lacks a hyperk\"{a}hler structure and has a Hitchin base whose dimension is large relative to the fibers. The fibers are still isotropic tori, so the moduli space gives an overdetermined integrable system. A moduli space of this type is called a Hitchin-type system. Moving to the parabolic perspective, prescribing the parabolic divisor restricts the spectral curves of the Higgs bundle. The moduli space of Higgs bundles with a given parabolic divisor cuts the dimension of the Hitchin base down to half that of the moduli space, which is once again a completely integrable system with hyperk\"{a}hler metric.

The utility of Higgs bundles arises partly from the various avatars of their moduli space. On the one hand, we can describe the locus of stable Higgs bundles via a set of differential equations, namely the Hitchin equations. A Hermitian metric $h$ on a vector bundle of degree $0$ is called \emph{harmonic} if 
\[F + [\phi, \phi^{*_h}] = 0 \]
where $F$ is the curvature of the Chern connection of $h$. A Higgs bundle admits a harmonic metric if and only if it is a direct sum of stable bundles with the same slope \cite{hitchin_self-duality_1987,donaldson_twisted_1987,corlette_flat_1988,simpson_naht_1991,simpson_loc_1992}. Every Higgs bundle with harmonic metric has an associated flat $GL(n,\CC)$ connection:
\[\nabla_\phi = \nabla_h + \phi + \phi^{*_h}.\]
The flatness of this connection follows from Hitchin's equations. This establishes the \textit{nonabelian Hodge correspondence} between flat $GL(n,\CC)$ connections and stable Higgs bundles, giving a complex version of the Narasimhan-Seshadri theorem. Parabolic Higgs bundles restricted to the complement of the parabolic divisor are ordinary Higgs bundles, and so they give flat connections on $\Sigma - D$. The flag and weights at a parabolic point define the monodromy of the flat connection around that point (cf. \cite{simpson_harmonic_1990}).

On the level of moduli spaces, this provides a diffeomorphism between the moduli of stable Higgs bundles $\MHiggs$ and the character variety $\Hom_{\irr}(\pi_1(\Sigma),GL(n,\CC))/GL(n,\CC)$. These have natural but non-isomorphic complex structures: $\MHiggs$ receives one from the complex structure on $\Sigma$, and the character variety inherits one from the complex structure on $GL(n,\CC)$. These together give two K\"{a}hler metrics, which combine to form the hyperk\"{a}hler metric on Hitchin moduli space.  For further details on the nonabelian Hodge correspondence, we refer the reader to \cite{simpson_I_1994,simpson_II_1994,garciaraboso_introduction_2015}.

\section{Higgs bundles as crystal moduli} \label{sec: Crystal Moduli}

We may now reach the primary purpose of this paper, which is to study the hyperbolic band theory of Section \ref{sec:hyperbolic_Band_Theory} using the Higgs bundle machinery exposited in Section \ref{sec:Higgs_Bundles}. Briefly, the spectral data of a Higgs bundle encodes a Riemann surface with a line bundle, equivalently a crystal lattice with abelian crystal momentum. The moduli space of Higgs bundles thus parametrizes a family of crystals. Abelianization passes from crystal data on the base curve to that on the spectral curve, trading complexity in the structure group ($U(n)$ vs. $U(1)$) with complexity in the underlying curve (the genus of the spectral curve is greater than that of the base). To use this in band theory, we need to relate the Hamiltonians on base and covering curve. Recalling the pushforward of a section of $L$ is a section of $p_*L$, we define the pushforward of the Hamiltonian $H$ on $L$ to be the operator $p_*H$ on $p_*L$ that makes the following diagram commute:

\[\begin{tikzcd}
	{H^0(\Sigma,L)} & {H^0(\Sigma,L)} \\
	{H^0(C,p_* L)} & {H^0(C,p_* L)}
	\arrow["H", from=1-1, to=1-2]
	\arrow["{p_*}"', from=1-1, to=2-1]
	\arrow["{p_*}", from=1-2, to=2-2]
	\arrow["{p_*H}"', from=2-1, to=2-2]
\end{tikzcd}\]
In this section we describe the pushforward of a crystal Hamiltonian, and relate it to Higgs bundles.

\subsection{Unbranched covers}
The pushforward is simplest without any branch points. For an unramified $N$-to-1  spectral cover $p:\Sigma \to C$, a degree zero line bundle $L \to \Sigma$ pushes forward to a degree zero vector bundle $E=p_*L \to C$.\footnote{Note that an unramified spectral cover comes from a Higgs bundle on $C$ valued in a degree zero bundle \cite{markman_spectral_nodate}.}
$L$ carries a flat connection $\nabla_L$, which pushes forward to a flat connection $\nabla_E$ on $E$. Specifically, $\nabla_E$ is the flat connection whose kernel is the pushforward of the sheaf of flat sections of $\nabla_L$. This operator exists on any open $U\sub C$ whose preimage $p\inv(U)$ contains $N$ connected components, and since there are no branch points it glues together into a global operator. For a section $\hat{\psi}$ of $L$ pushing forward to a section $\psi$ of $E$, we by definition have $\nabla_E \psi = p_*(\nabla_L \hat{\psi})$.

Consider a potential $\widehat{V}$ on $\Sigma$ symmetric under deck transforms. This is the lift of a potential $V$ on the base $C$. The abelian crystal Hamiltonian on $\Sigma$ is $H_L = \nabla_L^* \nabla_L + \widehat{V}$, and it pushes forward just as one would expect:

\begin{proposition}
For an unramified cover $p:\Sigma \to C$ with $p_*L = E$, the pushforward of $H_L = \nabla_L^*\nabla_L + \widehat{V}$ is $H_E = \nabla_E^*\nabla_E + V$
\end{proposition}

\begin{proof} 
Consider a section $\psi$ of $E$, and denote by $\hat{\psi}$ the section of $L$ satisfying $p_*\hat{\psi} = \psi$. We wish to show  
\[H_E \psi = p_* H_L \hat{\psi}.\]
First, we pushforward  the potential term. Since $\widehat{V}$ is constant along the fibers, \[p_*(\hat{\psi} \hat V) = p_*(\hat{\psi}  p^* V) = p_*(\hat{\psi)}   V = \psi V,\]
meaning the pushforward of $\widehat{V}$ is $V$.

Next we want the pushforward of $\nabla_L^* \nabla_L$. The adjoint is with respect a Hermitian metric on $E \otimes T^*C$, where the metric on $T^*C$ is induced by the Riemannian metric $g_C$ on $C$, and $E$ carries the flat metric $h_\flat$ induced by the standard Hermitian metric on $\HH \times \CC^n$. This is defined by
\[\langle \varphi \otimes \mathbf{\nu}, \psi \otimes \mathbf{\omega}\rangle_{E} = h_\flat(\varphi,\psi)\cdot g_C(\mathbb{\nu},\mathbb{\omega}),\]
where $\varphi,\hat{\varphi}$ are respective sections of $E$ and $L$ related by $p_*\hat{\varphi}=\varphi$, and $\nu,\omega$ are cotangent vectors.
The line bundle $L$ also carries a Hermitian metric $\hat{h_\flat}$ induced from the standard one on $\HH \times \CC$. Along with the Riemannian metric $g_\Sigma$, this defines a metric on $L\otimes T^*\Sigma$:
\[\langle \hat{\varphi} \otimes \hat{\nu}, \psi \otimes \hat{\omega}\rangle_{L} = \hat{h}_\flat(\hat{\varphi},\hat{\psi})\cdot g_\Sigma(\hat{\nu},\hat{\omega}).\]

The metric $h_\flat$ is the pushforward of $\hat{h}_\flat$, meaning $h_\flat(\varphi,\psi) = \hat{h}_\flat(\hat{\varphi},\hat{\psi})$. Furthermore, $g_C$ is the pushforward of $g_\Sigma$. Since there are no branch points, the pullback of $g_C$ to $\Sigma$ has constant negative curvature, just like $g_\Sigma$. By the uniformization theorem they must agree, meaning $p^*g_c = g_\Sigma$ and thus $g_c = p_*g_\Sigma$. Together, this implies that $\langle,\rangle_E$ is the pushforward of $\langle,\rangle_L$. The result now follows by juggling the definition of pushforwards and the adjoint:

\begin{align} 
    \langle \varphi, \nabla_E^*\nabla_E \psi \rangle_E &
    =\langle \nabla_E \varphi, \nabla_E \psi \rangle_{E}
    = \langle p_*{\nabla_L \hat{\varphi}}, p_*{\nabla_L \hat{\psi}} \rangle _{E}  \label{eq:pushforward_1} \\
    &= \langle \nabla_L \hat{\varphi}, \nabla_L \hat{\psi} \rangle _{L}
    = \langle  \hat{\varphi}, \nabla_L^* \nabla_L \hat{\psi} \rangle _{L}   \label{eq:pushforward_2} \\
    &= \langle  \varphi, p_* \nabla_L^* \nabla_L \hat{\psi} \rangle _{E}. \label{eq:pushforward_3}
\end{align}

We start over $C$ in equation \ref{eq:pushforward_1}, pull back to $\Sigma$ in \ref{eq:pushforward_2}, and push forward back to $C$ in \ref{eq:pushforward_3}. The identity in equations (\ref{eq:pushforward_1}-\ref{eq:pushforward_3}) holds for any choice of $\varphi$, so $\nabla^*_E\nabla_E \psi = p_*\nabla_L^* \nabla_L \hat{\psi}$. Together with the fact that $p_* \widehat{V} = V$, we conclude that $p_*(\nabla_L^*\nabla_L + \widehat{V})  = \nabla_E^*\nabla_E + V$ as desired. 
\end{proof}

In particular, if $\hat{\psi}$ is an eigensection of $H_L$ with eigenvalue $\lambda$, then
\[H_E \psi = p_*{H_L \hat{\psi}} = p_*{\lambda \hat{\psi}} = \lambda \psi,\]
and so $\lambda$ also an eigenvalue of $H_E$. Likewise, if $\varphi$ is an eigensection of $H_E$ with eigenvalue $\lambda'$, then $\hat{\varphi}$ is a eigensection of $H_L$ with the same eigenvalue.  Thus, the two spectra coincide. Since the associated eigensections are mapped to each other by pushforward, so the spectral data of $H_L$ and $H_E$ are equivalent. 

In effect, we trade a high-rank crystal momentum on the low genus curve $C$, with a rank $1$ crystal momentum on the higher genus $\Sigma$. For vector bundles on $C$ arising as pushforward of a line bundle on $\Sigma$, the band structure equals the abelian band structure over $\Jac(\Sigma)$.\footnote{Not all vector bundles on $C$ arise as pushforwards from $\Sigma$. This is clear by comparing the dimension of moduli space of the former ($n^2(g-1) + 1$) to that of the latter ($n(g-1)+1$). }
Interestingly, unbranched coverings arose in \cite{maciejko_automorphic_2021} as clusters. $N$-fold unbranched coverings of the unit cell can be built as a cluster of $N$ contiguous unit cells of $\Sigma$, with proper identification of the clusters' edges. This served to discretize the Jacobian of $C$. This section indicates the abelian band structure of the cluster gives part of the higher rank band structure on $C$.

\subsection{Branched covers} \label{sec:branched covers}

In most applications, holomorphic maps between Riemann surfaces are branched. Our crystal of interest is a branched cover $p:\Sigma \to C$ with branching locus $B\sub C$, and crystal momentum $L \to \Sigma$. We want the pushforward operator of $H_L = \nabla_L^* \nabla_L + \widehat{V}$. Since differential operators are local, it suffices to describe the pushforward in a neighborhood of each point. The complement of the branch points $C' = C \backslash B$ is an unbranched cover, and so the reasoning from last section shows that $H_L$ pushes forward to $H_E =  \nabla_E^* \nabla_E + \widehat{V}$. Here $\nabla_E$ is the pushforward of $\nabla_L$ on $C'$, and the adjoint is with respect to the constant negative curvature metric on $C'$ whose pullback smoothly extends to all $\Sigma$.  This glues together to a global operator $H_E$ on $L^2(C',E)$. We wish to extend $H_E$ to a self-adjoint operator $p_*H_L$ on $L^2(C,E)$, whose eigensections are exactly the pushforwards of those of $H_L$. 

To achieve this, we restrict the domain of $H_E$ from $L^2(C',E)$ to bounded sections. Indeed, if $\hat{\psi}$ is an eigensection of $H_L$, then elliptic regularity implies it is smooth and in particular bounded on $\Sigma$. Its pushforward $\psi$ is an eigensection of $H_E$ on $C'$, and is bounded on $C$. Conversely, consider a bounded eigensection $\psi$ of $H_E$ on $C'$, with eigenvalue $\lambda$. It is the pushforward of a bounded section $\hat{\psi}$ of $H_L$ on $\Sigma' = \Sigma \backslash p\inv(B)$. Extending $\hat{\psi}$ by zero to $p\inv(B)$ gives a weak solution to the equation $(H_L - \lambda)\hat{\psi}=0$. Elliptic regularity implies that $\hat{\psi}$ is in the $L^2$ class of an bonafide smooth eigenfunction $\hat{\psi}_{s}$. The pushforward $\psi_s = p_*\hat{\psi}_{s}$ is a smooth, bounded section of $E$ that agrees with $\psi$ on $C'$, since they are smooth and in the same $L^2$ class. That is, a bounded eigensection of $H_E$ on $C'$ comes from an eigensection of $H_L$ on $\Sigma$. So, the spectrum of $H_L$ is captured in the restriction of $H_E$ to bounded sections. It is the potential plus the canonical self-adjoint extension of Laplacian from $L^2(C',E)$ to $L^2(C,E)$ known as the Friedrichs extension (which exists because the Laplacian is non-negative and symmetric) \cite{kay_boundary_1991}.

The Friedrichs extension abstractly characterizes the pushforward operator, but is not very elucidating. To better understand the behavior at branch points, we treat the pushforward of the flat connection as a parabolic connection. For the sake of clarity, we turn to the simple case where $p:\Sigma \to C$ is a double cover. When the base curve is $\PP^1$, this is the hyperelliptic case described in Section \ref{sec:Hyperelliptic}. 
Let us first describe the pushforward connection. Transporting around a branch point locally interchanges the sheets of $\Sigma$, giving the involution $\sigma$. Following Section \ref{sec:parabolic}, the even and odd sections of $L$ push forward to a frame of $E$ on an open set surrounding the branch point. This monodromy yields the action of $\sigma$ on $E$, which in this frame is
\[\begin{pmatrix}
1 & 0\\0 & -1
\end{pmatrix}.\]
This comes from a logarithmic flat connection $p_*\nabla_L$ on $E$ with simple pole at the branch point. Its residue there is 
\[\Res \left(p_*\nabla_L \right) =
\begin{pmatrix}
0 & 0\\0 & \frac{1}{2}
\end{pmatrix}.\]
To summarize, the pushforward connection is a flat, parabolic connection $\nabla_E$ on $E$, with parabolic points of weights $(0,1/2)$ at each branch point and a distinguished line in $E$ consisting of even sections of $L$. 

To define the Hamiltonian we also need the adjoint $\nabla_E^*$, which involves both the Riemannian metric on $\Sigma$ and a Hermitian metric on $E$. The Riemannian metric is described in Section \ref{sec:parabolic}. It is the unique constant negative curvature metric with cone points of angle $\pi$ at every branch point. Next, the Hermitian metric on $E$ is the pushforward of the flat Hermitian metric $h_\flat$ on $L$, induced by the standard Hermitian metric on $\HH \times \CC$. This metric is singular: introducing a conformal coordinate $w$ on $C$ around a branch point, the double cover $\Sigma$ has coordinate $z$ satisfying $w=z^2$. The holomorphic sections of $L$ given by $1$ and $z$ are respectively even and odd under interchange of sheets, and push forward to a holomorphic frame $(e_e, e_o)$ of $E$. We see $p_*h_\flat(e_e,e_e) = h_\flat(1,1) = 1$, while $p_*h_\flat(e_0,e_0) = h_\flat(z,z) = z^2 = w$. This means $p_*h^\flat$ is singular along the distinguished line, vanishing to order $z$, so is adapted to the parabolic structure with the weights $(0,1/2)$. 

To summarize, the operator $\nabla_L^* \nabla_L$ pushes forward to $\nabla_E^*\nabla_E$. The parabolic connection $\nabla_E$ is adapted to the parabolic structure on E, as is the Riemannian metric $g_C$ and Hermitian metric $p_*h^\flat$ used to define the adjoint. The Hamiltonian associated to a rank 2 Higgs bundle $(E,\phi)$ is $H = \nabla_E^* \nabla_E + V$, where $E$ has parabolic points at the zeros and poles of $\det(\phi)$, and distinguished line defined by the matrix of $\phi$ at these points, which is either nilpotent or has nilpotent residue. We expect the extension to rank$>2$ Higgs bundles to follow a similar story.

\begin{remark}
For a hyperelliptic curve, this relates the line bundle on $\Sigma$ with a rank 2 parabolic vector bundle on $\PP^1$ and an associated logarithmic connection. Up to a shift of degree, it gives us a rank 2 Fuschian system. We write this explicitly for Euclidean crystals in Section \ref{sec:genus 1}.
\end{remark}

\subsection{Application to band theory}
By packaging the crystal lattice and abelian crystal momentum into the spectral data of a Higgs bundle, we can parametrize a family of crystal data with the moduli space of Higgs bundles $\MHiggs$. As described in Section \ref{sec:Moduli of Higgs bundles}, the spectral data gives a fibration, with base parametrizing spectral curves and fibers their Jacobians. For a hyperbolic crystal, the Hitchin base parametrizes crystal lattices, and the fibers parametrize abelian crystal momenta. For example, the moduli space of rank $2$ Higgs bundle on $\PP^1$ valued in $\calO(3)$ encodes all genus 2 crystals and all abelian crystal momenta (as every genus 2 curve is hyperelliptic).

Assigning a potential $V$ on $C$, each Higgs bundle defines a Hamiltonian described in Section \ref{sec:branched covers}. The graph of the spectrum of these Hamiltonians gives the band structure over $\MHiggs$. Restricted to any fiber, this is the rank one band structure for that curve. We have constructed a sort of moduli space of band structures (illustrated in Figure \ref{fig:universal_band_structure}). More suggestively, we could frame this as a universal object. On $\MHiggs \times C$, there is a universal Hamiltonian which restricts on $\{(E,\phi)\} \times C$ to the Hamiltonian defined above. The band structure is then a (possibly singular) analytic submanifold of $\MHiggs\times\CC \times C$ derived from the universal Hamiltonian, constant along $C$. This point of view transforms hyperbolic band theory into a moduli problem. We describe this band structure in detail for Euclidean crystals in Section \ref{sec:genus 1}. Moreover, the level crossings on each fiber (often occurring on high-symmetry momenta) glues together to a level crossing set on all $\MHiggs$ (see Figure \ref{fig:universal_band_structure}). We speculate more on the branching of band structure along high symmetry branes on $\MHiggs$ in Section \ref{sec:Speculation/Symmetry}.

In some limits, the space of possible band structures is finite-dimensional. For example, the tight-binding limit gives a finite-dimensional Hilbert space, thus there is a finite-dimensional moduli space $S$ of all tight-binding Hamiltonians. (These are described in Section \ref{sec:Tight binding}.) The Hamiltonian depends on a point in $\MHiggs$, giving some action of $\MHiggs$ on $S$. In any case, all band structures can be encoded in a finite-dimensional space $S \times \MHiggs$.

\begin{figure}[htp]
    \centering
    \includegraphics[width = \textwidth]{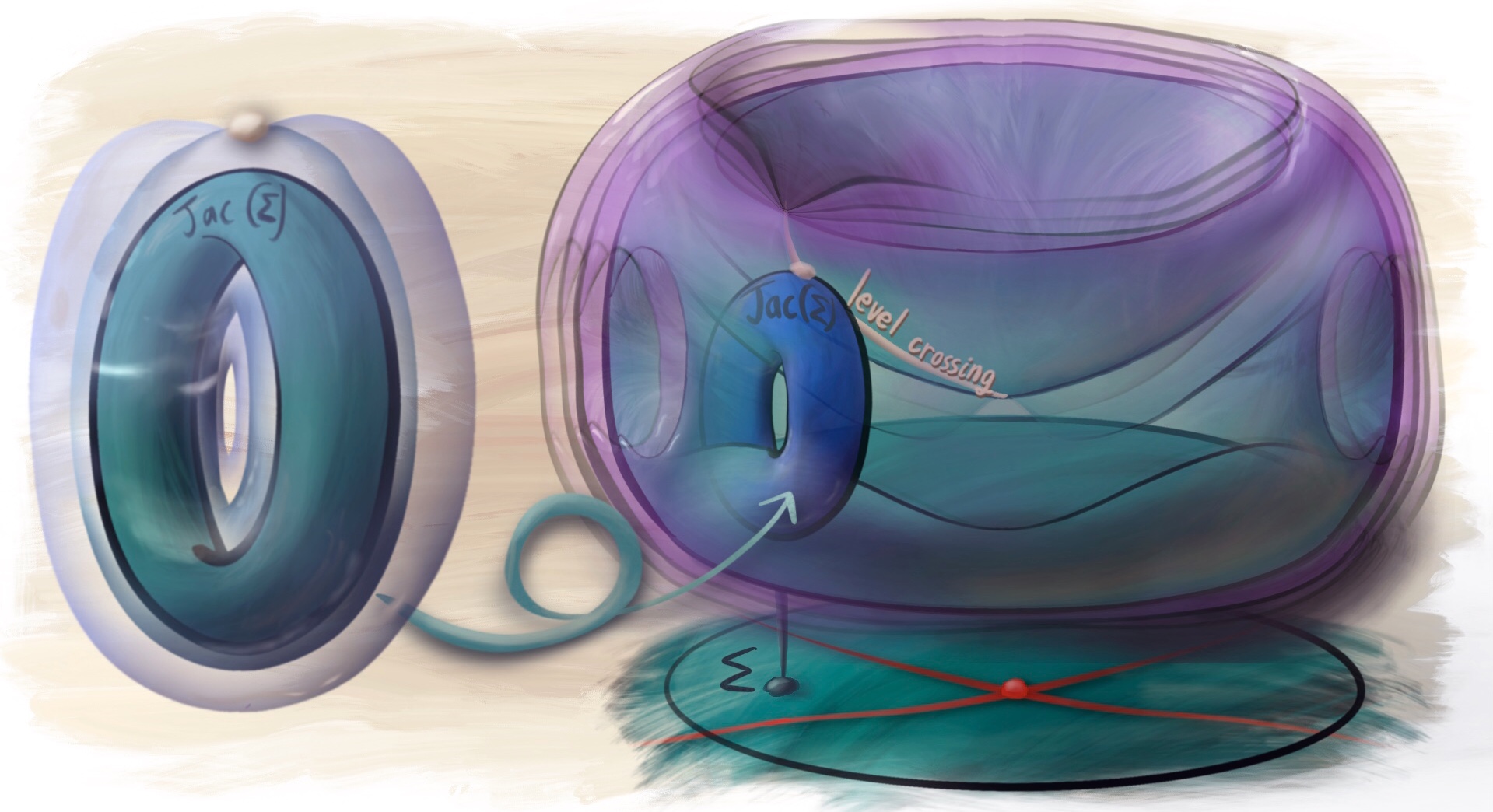}
    \caption{The universal band structure over the moduli space of Higgs bundles. Each point on the base gives a spectral curve, representing a hyperbolic crystal. The Hamiltonian defines a band structure over its fiber, the Jacobian of the spectral curve. These band structures glue together to form a master band structure on all $\MHiggs$, represented by the translucent purple shells. The level crossings on each fiber, drawn as a point over $\Jac(\Sigma)$, glue together to a level crossing line, indicated on the figure of $\MHiggs$.}
    \label{fig:universal_band_structure}
\end{figure}

\section{Higgs bundles as complex momenta}
In the preceding section, we used the spectral data of a Higgs field to parametrize different crystals, which we designate the \textit{crystal moduli interpretation}. Higgs fields also arise naturally in a rather different context: they define the imaginary part of the crystal momentum, which we call the \textit{complex momentum interpretation}. To motivate this, the spectral curve of a Higgs bundle is tantalizingly similar to a band structure. (In fact, this observation was one of the motivations for developing hyperbolic band theory in \cite{maciejko_hyperbolic_2021}.) However, techniques from (complex) algebraic geometry do not readily apply to the real eigenvalues of the Hamiltonian, and so we need to complexify. The Hamiltonian is self-adjoint because the crystal momenta are unitary representations $\pi_1(\Sigma) \to U(n)$. Allowing non-unitary crystal momenta $\pi_1(\Sigma) \to GL(n,\CC)$ permits non-self-adjoint Hamiltonians and, thus, complex energies. Just as the $U(n)$ representations parametrize holomorphic vector bundles, $GL(n,\CC)$ representations parametrize Higgs bundles through the nonabelian Hodge correspondence.

\subsection{Complex crystal momenta}

Let us once again look to the Euclidean case for inspiration. A crystal momentum is a unitary character of the lattice $\chi_k:\Gamma \to U(1)$, such that translation by a lattice vector $\gamma$ multiplies the phase by $\chi_k(\gamma)$. The crystal momentum is determined by a vector $k \in \RR^2$, with associated Hamiltonian $H_k = ( \nabla + ik)\cdot (\nabla + ik) + V$. A \textit{complex} momentum is a character $\chi_k:\Gamma \to \CC^*$, so that lattice translation changes both the phase and the amplitude. The space of complex crystal momentum is not $U(1)^2$, but ${\CC^*} ^2$. This is realized by allowing $k \in \CC^2$ in the Hamiltonian $H_k$. Indeed, the Hermitian conjugate of $H_k$ is
\[H_k^\dagger =(-\nabla - i\bar{k})\cdot (- \nabla- i\bar{k}) + V = (\nabla + i\bar{k})\cdot (\nabla+ i\bar{k}) + V = H_{\bar{k}},\]
and $H_k$ is only self adjoint when $k$ is purely real. Non-self adjoint Hamiltonians are not standard in physics because the amplitude of an eigenstate can change over time, but we ignore this and focus on the mathematical consequences, at least for the moment.

Abelian crystal momenta on a hyperbolic crystals are complexified in much the same way. We replace unitary representations $\pi_1(\Sigma) \to U(1)$ with general ones $\pi_1(\Sigma) \to \CC^*$. Decomposing $\CC$ into phase $U(1)$ and scaling $\RR^{>0}$ splits the character variety
\[\Hom(\pi_1(\Sigma),\CC^*) \cong \Hom(\pi_1(\Sigma),U(1)) \times \Hom(\pi_1(\Sigma),\RR^{>0}).\]
The quotient by the conjugation action was ignored since $\CC^*$ is abelian. The first factor is the Jacobian $U(1)^{2g}$, and the second is $(\RR^{>0})^{2g} \cong \CC^g$. A point in this character variety uniquely defines a flat connection $D_{(L,\phi)} = \de + k$ on the trivial line bundle $L\to \Sigma$, where $k$ is a harmonic one-form. The splitting of the character variety is induced by the decomposition $k = k_r + i k_i$, with real part $k_r = \frac{1}{2}(k + \bar{k})$ and imaginary part $k_i = \frac{1}{2i}(k - \bar{k})$. $k_r$ controls the $U(1)$ part of the monodromy, while $k_i$ controls the $\RR^{>0}$ part. The associated Higgs field is the unique holomorphic one-form $\phi$ such that $k_r = \phi + \bar{\phi}$. Likewise, every rank $1$ Higgs bundle $(L,\phi)$ gives a flat connection $D = \nabla_L + i(\phi + \bar{\phi})$, formed from the unitary $\nabla_L$ with an added skew-Hermitian part derived from the Higgs field. This fact follows from the Hodge decomposition $H^1 \cong H^{0,1} \oplus H^{1,0}$, and is the abelian prototype of the nonabelian Hodge correspondence. See \cite{goldman_rank_2008} for more details on rank $1$ Higgs bundles.

The story is similar for higher rank crystal momenta. The Hamiltonian again has the form 
\[H_E = \nabla_E^*\nabla_E + V = (\de + iA)^*(\de +iA) + V,\]
where $A$ is now a $\mathfrak{u}(n)$-valued 1-form. To complexify, we replace $\mathfrak{u}(n)$ with $\mathfrak{u}(n) \otimes \CC \cong \mathfrak{gl}(n,\CC)$. Any matrix $M \in \mathfrak{gl}(n,\CC)$ splits into a Hermitian part $(M+ M^\dagger)/2$ and a skew-Hermitian part $(M- M^\dagger)/2i$.  Thus, $\mathfrak{gl}(n,\CC)$ splits into $\mathfrak{u}(n) \oplus \mathfrak{u}(n)$, where the first factor contains Hermitian matrices and the second skew Hermitian. Denoting the Hermitian conjugate with respect to the Hermitian metric $h$ by $\dagger_h$, $A$ splits into Hermitian part $A_I = \frac{1}{2}(A+A^{\dagger_h})$ and skew Hermitian part $A_R = \frac{1}{2}(A-A^{\dagger_h})$. The Hermitian adjoint of the Hamiltonian is
\[H_E^\dagger = (-\de  - iA^{\dagger_h})^*(-\de  - iA^{\dagger_h}) + V = (\de  + iA^{\dagger_h})^*(\de  + iA^{\dagger_h}) + V,\]
and so $H_E$ is self adjoint whenever $A=A^{\dagger_h}$, or $A_R = 0$. For a Higgs field $\phi$ on $E$, the associated flat connection is $D = \nabla_E + \phi + \phi^{\dagger_h}$ where $\nabla_E$ is the flat connection from the holomorphic structure of $E$. Conversely, any skew Hermitian one-form $A_R$ has a unique holomorphic $\End(E)$-valued one-form $\phi$ such that $i(\phi +\phi^{\dagger_h}) = A_R$. For a Higgs bundle $(E,\phi)$, $E$ controls the Hermitian part and $\phi$ the skew-Hermitian part of the connection, reflecting the splitting $\mathfrak{gl}(n,\CC) \cong \mathfrak{u}(n) \oplus \mathfrak{u}(n)$.

The crystal Hamiltonian associated to a Higgs bundle $(E,\phi)$ is $H = D_\phi^* D_\phi + V$ for flat connection $D_\phi = \nabla_E + i(\phi + \phi^{\dagger_h})$. Once again, we see $\phi$ controls the complex part of the momentum. The Hamiltonian is self-adjoint only when $\phi=0$. For parabolic bundles, the Hamiltonian is only defined on the complement of the parabolic locus, but following Section \ref{sec: Crystal Moduli}, the Friedrichs extension gives an operator on $L_2(\Sigma)$

\begin{remark}
In complexifying the Hamiltonian, we chose the adjoint to be complex linear. We could have chose it to be complex antilinear (a Hermitian conjugate), which would make the Hamiltonian manifestly self-adjoint. We chose the first complexification for the physically motivated non-real spectrum. The resulting Hamiltonian $H_\rho$ varies holomorphically with the irreducible representation $\rho$, with respect to the natural complex structure on the character variety.
\end{remark}

\subsection{Bloch variety and tight binding models}
\label{sec:Tight binding}
The Hamiltonian $H_\rho$ for each crystal momentum $\rho:\pi_1(\Sigma)\to GL(n,\CC)$ is defined through the associated Higgs field $(E,\phi)$, via $H_\rho = D_\phi^* D_\phi + V$. We can build a band structure over the moduli space of representations (diffeomorphic to $\MHiggs$) by graphing the spectrum of $H_\rho$. More specifically, we can assemble the eigenstates of $H_\rho$ into a master Bloch function $\psi(z,k,e)$ satisfying
\[H_\rho \psi(z,\rho,e) = e \psi(z,\rho,e).\]
The band structure is the set of pairs $\rho \in \MHiggs$, $e \in \CC$ where this equation has a solution. We call this set the Bloch variety $\bloch \sub \MHiggs \times \CC$. This is also the zero locus of the equation $\det(H-E)$ (for some suitably regularized determinant) making $\bloch$ is a codimension $1$ analytic (possibly singular) submanifold in $\MHiggs \times \CC$. Note that $\bloch$ is a variety with respect to the complex structure of the representation variety because $H_\rho$ (and thus its spectrum) varies holomorphically with respect to $\rho$.

\begin{remark} 
Physics expects some properties of the band structure, such as the ``avoided crossing" phenomena, where a perturbation lifts degeneracies at level crossings. The Von Neumann-Wigner theorem formalizes avoided crossing, stating that generic band structures have codimension $2$ level crossings \cite{von_neuman_uber_1929}. This theorem is usually evoked for Hermitian operators with real energies, but the same reasoning applies to generic operators with complex energies, where branching is complex codimension 2. The purity of the branch locus implies a holomorphic map between nonsingular algebraic varieties branches along a codimension $1$ subvariety \cite{zariski_purity_1958}. The codimension 2 crossings must therefore be singularities of the Bloch variety. This is apparent in the canonical example of a perturbatively stable level crossing, the Dirac point of graphene, where the crossing is a conical singularity. 
\end{remark}

The Bloch variety is especially tractable when it forms a finite branched cover of crystal momentum space. This occurs for finite-dimensional Hamiltonians, such as the tight-binding model from condensed matter physics.\footnote{One can also get a finite dimensional Hamiltonian by discretizing space and using a discrete Laplacian, see \cite{gieseker_geometry_1993}. This approach is applicable in different regimes than the tight-binding model. }
In its simplest iteration, each crystal cell gets a finite-dimensional vector space of states, and we assume a cell only interacts with its neighbors (as if each state is ``tightly bound" to a given cell).
The Hamiltonian for each cell splits into an \textit{on-site matrix} $M$ that couples the cell to itself, \textit{hopping matrices} $J$ that couple the cell to its neighbors. These models are a simple way to realize quantum materials experimentally or numerically. For example, hyperbolic crystals were experimentally realized by constructing a tight binding model with sites connected in a $\{3,7\}$ hyperbolic tiling \cite{kollar_hyperbolic_2019}. The theory of hyperbolic tight binding models was briefly treated in \cite{maciejko_hyperbolic_2021}, and further built upon in \cite{maciejko_automorphic_2021}. The connections between sites are described by a basis of cycles $\{A_i,B_i\}$ of $\pi_1(\Sigma)$, each with their own hopping matrix  $J_{A_i}, J_{B_i}$.  Hopping backwards is governed by $J_\gamma^\dagger$, and so for a trivial crystal momentum the cells Hamiltonian is
\[H = M + \sum_{\gamma \in A_i,B_i} J_{\gamma} +  J_{\gamma}^\dagger.\]
The full Hamiltonian is the sum of each cell's Hamiltonian, acting on the direct sum of every cell's vector space of states. For an abelian hyperbolic Bloch state with complex crystal momentum $k$ and associated representation $\chi_k:\Gamma \to \CC^*$, the modified Hamiltonian is 
\[H_k = M + \sum_{\gamma \in A_i,B_i} \chi_k(\gamma) J_{\gamma} + \chi_k(\gamma)^{-1} J_{\gamma}^\dagger.\]
We can see $H_k$ is a linear function of $\chi_k,\chi_k^{-1}$ for $2g$ parameters $\chi_k \in \CC^*$. The Bloch locus is the zero set of the characteristic polynomial $\det(H_k - E)$. In particular, this is the solutions of a finite degree polynomial in $\chi$ and $\chi^{-1}$, and so the Bloch locus is algebraic.

\subsubsection{Tight binding limit, multi-atomic crystals, and quivers}

The tight-binding limit produces a finite-dimensional approximation to a kinetic plus potential Hamiltonian in the limit of infinitely deep potential wells. The eigenstates localize around each well (thought of as an atom), and the only significant interaction is between neighboring atoms. The relevant part of Hilbert space is a finite-dimensional subspace of bound states, giving a tight-binding model. 
Going one step further, we split the vector space of the cell into vector spaces attached to each ``atom". Whenever the Hamiltonian has a nonzero matrix element between two atomic vector spaces, these atoms have a ``bond". We can associate a Hamiltonian on a configuration of atoms to a quiver with representation. Each atom in the unit cell gives a node attached to corresponding atomic vector space. Bonded atoms are connected by arrows in the quiver. We must also include bonds passing between adjacent unit cells, which embeds the quiver onto $\Sigma$. (An example quiver from an Euclidean crystal is shown in Figure \ref{fig:quiver}.) The Hamiltonian gives a representation of this quiver. 

A nonzero abelian complex crystal momentum $k$ modifies the Hamiltonian and thus changes the quiver representation. For each arrow passing from the unit cell to its neighbor in the direction $\gamma$, the associated hopping matrix is multiplied by $\chi_k(\gamma)$. Through this, $(\CC^*)^{2g}$ acts on the moduli space of quiver representations. The Bloch variety lives over an orbit of this group, and is always algebraic. For a nonabelian crystal momentum $\rho: \Gamma \to GL(n,\CC)$, the Hamiltonian lives on the direct sum of $n$ copies of the quiver, intertwined by $\rho$ when passing between unit cells. Quivers and their representations often appear in algebraic geometry and high-energy physics, so tight-binding models could fruitfully bridge these fields with hyperbolic crystallography. We speculate on one such possible relationship in Section \ref{sec:Speculation/High energy}. 
\begin{figure}[htp]
    \centering
    \includegraphics[width = \textwidth]{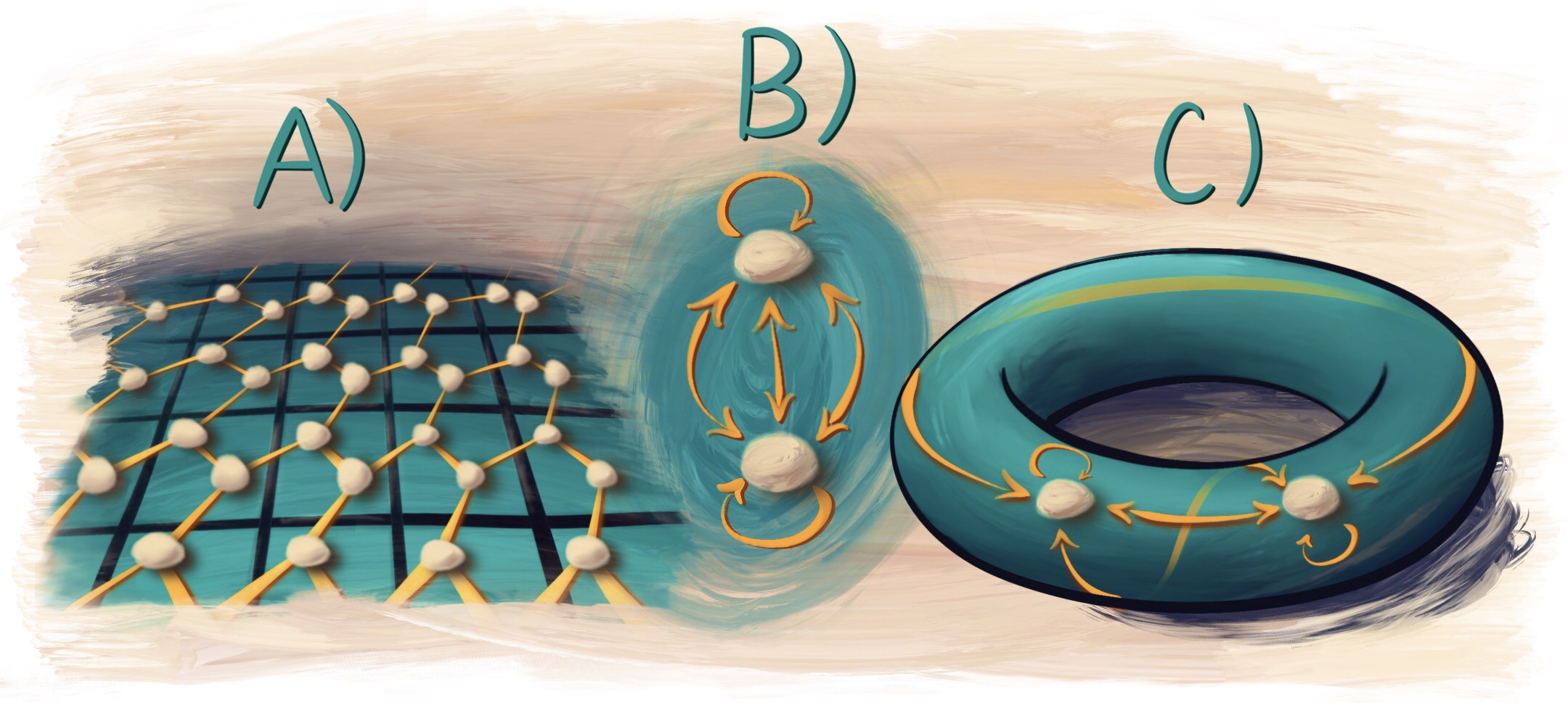}
    \caption{Illustration of the quiver associated with a crystal. A) A Euclidean crystal, with atoms at the white dots and bonds along the yellow lines. B) The associated quiver, which includes 3 double-sided arrows between the two nodes, and an arrow connecting each node to itself. C) The quiver, embedded on the torus associated to the Euclidean crystal.}
    \label{fig:quiver}
\end{figure}

\subsection{Crystal moduli interpretation vs. complex momentum interpretation}
We have encountered two natural interpretations of Higgs bundles within hyperbolic band theory. The crystal moduli interpretation, detailed in Section \ref{sec: Crystal Moduli}, treats the spectral curve of a Higgs bundle as the crystal unit cell and the associated spectral line bundle as an abelian crystal momentum. Then, the moduli of Higgs bundles gives a moduli space of crystal data.  Alternatively, the complex momentum interpretation treats Higgs bundles as the imaginary part of the crystal momentum. Following Simpson's terminology \cite{simpson_I_1994}, we could say the crystal moduli interpretation follows the ``Dolbeault'' picture, while the complex momentum interpretation is in the ``Betti'' picture. Ultimately, these interpretations give two crystal Hamiltonians for a Higgs bundle $(E,\phi)$ and potential $V$. In the complex momentum interpretation, we get the potential plus the Laplacian of the flat connection associated with the Higgs field:
\[H_{\textrm{complex}} = (\nabla_E+ \phi + \phi^\dagger)^* (\nabla_E + \phi + \phi^\dagger) + V.\]

In the crystal moduli interpretation, the Hamiltonian is the lift of the potential plus the Laplacian of the flat connection on the spectral line bundle over the spectral curve
\[H_{\textrm{crystal}}  = \nabla_L^* \nabla_L + \widehat{V}.\]
This pushes forward to to the self-adjoint extension of the operator 
\[H_{\textrm{crystal}}  = \nabla_{(E,P)}^* \nabla_{(E,P)} + V,\]
where we emphasize that $\nabla_{(E,P)}$ is the flat connection of a parabolic bundle with underlying bundle $E$ and parabolic data $P$. 
These two Hamiltonians on $C$ each give a different band structure. For instance, only the crystal moduli interpretation has real energies. The relationship between these two Hamiltonians is still unclear.

\subsection{Bulk-boundary correspondence}
\label{sec:bulk-boundary}
While complex crystal momenta do not give physical Bloch states in the bulk, they can describe edge states. They give an eigenstate which exponentially decays in the bulk, while the non-normalizable exponential growth is cut off by the boundary. An edge state is described by a map from each point of the boundary to a Bloch eigenstate with a complex energy. When the boundary is a circle, the topologically distinct edge states are homotopy classes of maps from the circle to the Bloch variety, i.e $\pi_1(\bloch)$. This was first discussed for the fractional quantum Hall effect on Euclidean crystals in \cite{hatsugai_chern_1993, hatsugai_edge_1993, dwivedi_bulk_2016}. Specifically, they demanded periodicity in one direction, reducing to a one-dimensional momentum space whose complex Bloch variety is a Riemann surface. The bulk-boundary correspondence relates a topological invariant of the bulk insulator (the first Chern number of the Bloch eigenbundle, also called the Hall conductance) with an invariant of a surface state (the winding number about the Fermi energy in the complex Bloch variety).

Hyperbolic crystals should exhibit a similar bulk-boundary correspondence, where the complex momenta are represented by Higgs bundles. For a section of the hyperbolic disc bounded by a circle, the topologically protected edge states should correspond to non contractible cycles in the Bloch variety over Hitchin moduli space. Hyperbolic crystals have topological invariants for bulk insulators similar to Euclidean crystals (Section \ref{sec:Speculation/Topological materials}), though the relationship with the edge states is still unclear.

\section{Euclidean crystals through Higgs bundles} \label{sec:genus 1}
Let us see this formalism in action and apply it to familiar 2D Euclidean crystals. We will find that Higgs bundles consolidate
the standard band theory canon into an algebraic geometry package. 
    
A Euclidean crystal is defined by a lattice $\Gamma = \langle 1,\tau \rangle \sub \CC$, whose unit cell $ \Sigma = \CC/\Gamma$ is a genus one Riemann surface. As discussed in Section \ref{sec: Euclidian bloch thm}, a periodic potential and the standard flat metric on $\CC$ define the Hamiltonian $ \Delta + V$. The abelian Bloch states are classified by a flat line bundle $L \to \Sigma$, and the Hamiltonian is $H_L = \nabla_L^* \nabla_L + V$ acting on sections of $L$. The band structure is the spectrum of $H_L$ as $L$ varies over the moduli space of flat line bundles, $\Jac(\Sigma)$. Since the genus is 1,  $\Jac(\Sigma)$ is isomorphic to $\Sigma$. The Jacobian is a group with distinguished identity, making it an elliptic curve. 

We wish to understand the band structure using Higgs bundles. To start, we describe the Higgs bundle associated with $\Sigma$ by the crystal moduli interpretation. Genus $1$ curves are the prototypical example of hyperelliptic curves, and so we follow Section \ref{sec:Hyperelliptic}. The branched covering $\Sigma \to \PP^1$ is defined by the equation $\lambda^2 = P(z)$, for a degree $2g+2 = 4$ polynomial $P(z)$. The branch points are at the roots of $P(z)$, which we take to be  $0,1,\infty,$ and $m$ after a M\"{o}bius transform. The location of $m$ uniquely determines $\Sigma$, and is related to the lattice $\langle 1,\tau \rangle$ through the modular lambda function,  $\lambda(\tau) = m$.

A Higgs field with genus 1 spectral curve lives on a rank $2$ bundle $E$ over $\PP^1$, valued in the line bundle $K(D)$ for a degree $g+3 = 4$ divisor $D$. Alternatively, these are strongly parabolic Higgs fields with 4 parabolic points. The parabolic points exhaust the $2g+2$ branch points (a property unique to genus $1$ spectral curves), and therefore the spectral curve is determined by $D$. The points in $D$ are the roots of $P(z)$, and thus the Higgs field has determinant proportional to
\[\det(\phi) \propto \frac{\de z^2}{z (z-1) (z-m)},\]
which is the meromorphic quadratic differential defining the spectral cover.

Next, consider the abelian crystal momentum defined by a line bundle $L$ over $\Sigma$. As a holomorphic bundle, this pushes forward to a degree $-2$ holomorphic vector bundle $E$ on $\PP^1$. Stability implies $E$ is either $\calO(-1)\oplus \calO(-1)$ or $\calO \oplus \calO(-2)$. The trivial line bundle must push forward to a bundle with a nonzero holomorphic section, which is $\calO \oplus \calO(-2)$. Every other line bundle pushes forward to $\calO(-1)\oplus \calO(-1)$. This vector bundle gets a parabolic structure from the elliptic involution, with parabolic points at the branch locus $\{0,1,\infty,m\}$. The parabolic structure of $E$ and associated meromorphic connection follow from Section \ref{sec:branched covers}. Uniquely to this situation, the parabolic structure of the Higgs field and the parabolic structure induced on $E$ by the Higgs field agree, since their parabolic divisors coincide. Consequently, every Higgs field with a given parabolic structure gives the same line bundle. The space of Higgs bundles with a given parabolic structure is a one-dimensional vector space, and the Higgs bundle has no control over the spectral data.

We can convert the parabolic connection on $E = \calO(-1)\oplus \calO(-1)$ to a pair of coupled meromorphic differential equations on $\CC$, also called a Fuschian system. These are equivalent to a flat connection on a trivial rank $2$ bundle over $\PP^1 \backslash D$, with prescribed monodromies around points in $D$. The trivial bundle is $\calO \oplus \calO \cong E \otimes \calO(1)$. To keep the parabolic degree zero we change the weights from $(0,\frac{1}{2})$ to $(-\frac{1}{4},\frac{1}{4})$. This changes the monodromy of the flat connection around each point by
\[\begin{pmatrix} 
0&-1\\
1&0
\end{pmatrix}
\to
\begin{pmatrix} 
i&0\\
0&-i
\end{pmatrix}.\]
The Fuschian system admits an explicit description \cite{heller_abelianization_2016}. Since $E$ is trivial, the eigenlines of each parabolic point define four lines in $\CC^2$, equivalently four points in $\PP^1$. After fixing the parabolic divisor and weights, the moduli of parabolic structures is isomorphic to 4 marked points on $\PP^1$, which equals $\PP^1$. Fix this structure with a number $u \in \PP^1 \cong \CC \cup \{\infty\}$. The associated connection $\nabla^u$ and the Higgs field $\phi^u$ are explicitly given by 
\begin{equation}
 \begin{gathered}
\nabla^u = \de + A_0^u \frac{\de z}{z} + A_1^u \frac{\de z}{z-1} + A_m^u \frac{\de z}{z-m}\\
A_0^u = \frac{1}{4}\begin{pmatrix}
-1 & 0 \\ -1 & 1
\end{pmatrix},\qquad
A_1^u = \frac{1}{4} \begin{pmatrix}
0 & 1 \\ 1 & 0
\end{pmatrix}, \qquad A_m^u = \frac{1}{4} \begin{pmatrix}
-1 & 2u \\ 0 & 1
\end{pmatrix}\\
\Phi^u = \Phi_0^u \frac{\de z}{z} + \Phi_1^u \frac{\de z}{z-1} + \Phi_m^u \frac{\de z}{z-m}\\
\Phi_0^u = \begin{pmatrix}
0 & 0 \\ 1-u & 0
\end{pmatrix},\qquad
\Phi_1^u =  \begin{pmatrix}
u & -u \\ u & -u
\end{pmatrix}, \qquad
\Phi_m^u =  \begin{pmatrix}
-u & u^2 \\ -1 & u
\end{pmatrix}.
\end{gathered}   
\label{eq:explicit_higgs}
\end{equation}

The parameter $m$ controls the parabolic divisor and thus spectral curve (crystal lattice), while the parameter $u$ controls the flag data and thus line bundle on the torus (crystal momentum). 

\begin{remark}
With equation \ref{eq:explicit_higgs}, we can  find the spectrum of the crystal Hamiltonian through the Laplacian of a 2\textsuperscript{nd}-order rank 2 meromorphic differential operator on $\PP^1$. While this is not helpful for Euclidean crystals, the spectrum of every hyperelliptic crystal can be found with a similar Fuschian system (with more poles). This could enable numerical calculations of the band structure for arbitrary genus hyperelliptic curves, avoiding the difficulties encountered in \cite{maciejko_hyperbolic_2021} from meshing each curve.
\end{remark}

\subsection{Description of moduli space}\label{sec:genus 1/moduli space}
The moduli space of these Higgs fields with specified parabolic divisor is a Hitchin system \cite[\S8]{fredrickson_asymptotic_2020}. In the parametrization of equations \ref{eq:explicit_higgs}, this consists of Higgs bundles $B\Phi^u$ for some $B \in \CC$.
The Hitchin base consists of meromorphic quadratic differentials
\[\calB = \left\{ \frac{B}{z(z-1)(z-m)} \de z^2 \bigg | B \in \CC\right\}  \cong \CC\]
and fibers are the Jacobian of the specified elliptic curve. This base does \textit{not} parametrize different lattices, which are already determined by the parabolic divisor. The base and fiber are one-dimensional, making the moduli space a complex surface. This is the simplest nontrivial Hitchin system, earning the moniker ``toy model''
$\Mtoy$ \cite{hausel_global_2013}. We sketch this in Figure \ref{fig:toy}.

\begin{figure}[htp]
    \centering
    \includegraphics[width = \textwidth]{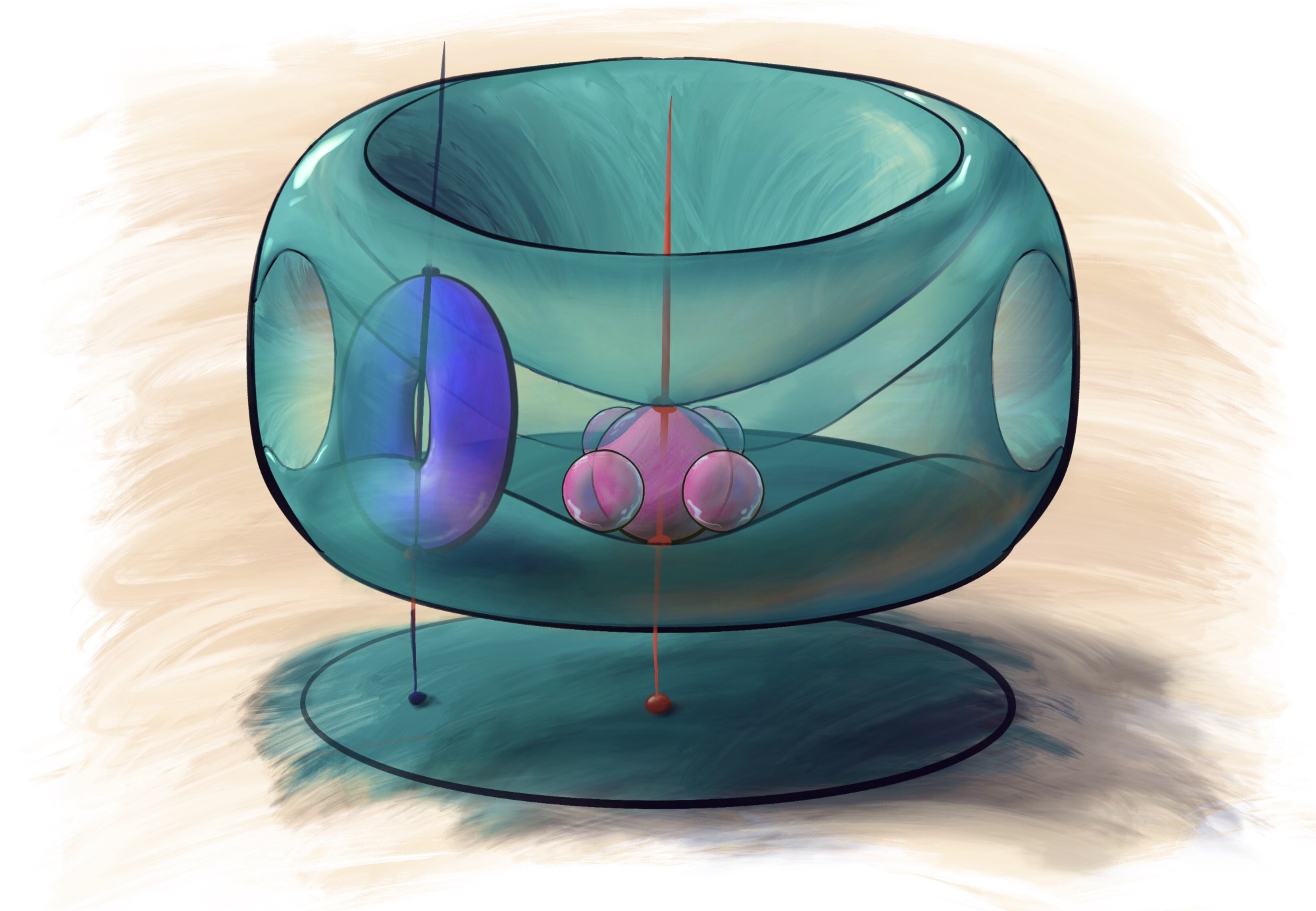}
    \caption{The moduli space of Higgs bundles on $\PP^1$ with 4 parabolic points, or Hausel's toy model $\Mtoy$. Every nonzero fiber is the Jacobian of a fixed torus (purple). The zero fiber, or nilpotent cone, is $\PP^1$ with 4 satellite copies of $\PP^1$ (pink). The Hitchin fibration maps it to $\CC$, and the singular locus is only $0 \in \CC$ (red dot). The internal lines illustrate the locus of $SL(2,\RR)$ Higgs fields.}
    \label{fig:toy}
\end{figure}

The total space $\Mtoy$ admits an explicit description. Start with $\CC \times \Jac(\Sigma) / \langle \pm 1\rangle$, where $\langle \pm 1\rangle$ multiplies the Higgs field by $\pm1$ and acts by $L \to L\inv$ on $\Jac(\Sigma)$. The quotient is singular at fixed points of this action, occurring where Higgs field is zero and the line bundle is a 2-torsion point of $\Jac(\Sigma)$. $\Mtoy$ is the resolution of singularities obtained by blowing up these 4 points. The resulting moduli space has regular fibers everywhere on the base $\calB$ except zero, where the fiber is a central $\PP^1$ with four satellite $\PP^1$s. These are located at $0,1,\infty$ and $m$ in the central $\PP^1$. This gives the nilpotent cone, shown in pink in Figure \ref{fig:toy}.

By forgetting the Higgs field, we send $\Mtoy$ to the moduli space of parabolic structures with given parabolic divisor on the underlying vector bundle $E$. This is determined by the four eigenlines at the parabolic points, which define the point $u \in \PP^1$.  The moduli space of parabolic structures is the central $\PP^1$ in the zero fiber of the Hitchin map of $\Mtoy$.

We can split $\Mtoy$ into strata corresponding to the underlying  holomorphic bundle.  Those with $E = \calO\oplus \calO(-2)$ come from the trivial line bundle on each Jacobian, and defines the ``small stratum" \cite{fredrickson_asymptotic_2020}. The rest have $E = \calO(-1) \oplus \calO(-1)$, defining the ``large stratum". Following equation \ref{eq: Hyperelliptic Higgs}, a Higgs bundle in the small stratum lives in
\[\phi \in H^0 \left(\Sigma, \begin{pmatrix} 
\calO(2) & \calO(4)\\
\calO & \calO(2)
\end{pmatrix}\right)\]
and every Higgs field is conjugate to 
\[\phi \in \begin{pmatrix} 
0 & P(z)\\
1 & 0
\end{pmatrix}.\]
The small stratum is the twisted analogue of the Hitchin section, which coincides with the theta divisor on each Jacobian. Figure \ref{fig:toy} visualizes this as the base dark blue circle, which is the graph of the identity of each Jacobian. Visualizing the Jacobian as a vertical torus, we take the identity to be the bottom point, and the 2-torsion points to be the points in line with the bottom (see the line bisecting the purple torus in Figure \ref{fig:toy}.) The small stratum is the trace of the base of the torus, visualized in Figure \ref{fig:toy} by the dark blue circle in the total space $\Mtoy$. The trace of the 2-torsion points picks out the locus of $SL(2,\RR)$, which form a brane in $\Mtoy$. These are shown by the dark internal lines in \ref{fig:toy}. Note that the Hitchin fibrations carries the Gauss-Manin connection, which has monodromy around the central fiber giving the elliptic involution. So, the $2$-torsion points are fixed, and there are indeed $4$ connected components of the moduli of $SL(2,\RR)$ Higgs bundle. We further discuss the relationship between high symmetry points in the Jacobian and special Higgs bundles in Section \ref{sec:Speculation/Symmetry}. Finally, the tori tracing $\Mtoy$ in Figure \ref{fig:toy} grow larger away from the nilpotent cone, reflecting the hyperk\"aler metric.

Alternatively, we can study the moduli space of twisted Higgs bundles, which lets the parabolic divisor vary. The twisting line bundle $K(D) = \calO(2)$ is dual to the canonical bundle $K = \calO(-2)$, bestowing it the nickname ``co-Higgs field".
The moduli space of such bundles was studied in \cite{rayan_co-higgs_2013}. The Hitchin base is
\[\calB = H^0(\PP^1,\calO(2))\oplus  H^0(\PP^1,\calO(2)^{\otimes 2}) \cong \CC^8.\]
This moduli space describes every Euclidean crystal lattice, but is quite redundant. We want a moduli space between the two extremes, giving all Euclidean lattices and line bundles without redundant dimensions. To achieve this, we only specify the locations of $3$ parabolic points ($0,1$ and $\infty$) and let the fourth vary. This amounts to replacing $\Mtoy$ with a fibration over $\Mtoy$ with fibers corresponding to the moduli space of elliptic curves $\Mell$.

Let us review the form of $\Mell$. An elliptic curve is defined by a lattice $\langle 1,\tau \rangle$ in $\CC$ for a parameter $\tau$ in the upper half-plane $\HH$. The moduli space of lattices is the quotient of $\HH$ by the action of the modular group $SL(2,\ZZ)$. Topologically $\Mell$ equals $\PP^1$ with one puncture, but it has an interesting orbifold structure. The orbifold group of a point is the isomorphism group of the corresponding lattice. First, every lattice has $\ZZ_2$ inversion symmetry about its center, corresponding to the elliptic involution. When $\tau$ has zero real part, it defines a rectangular lattice with extra $\ZZ_2$ symmetry, so these are $\ZZ_2$ orbifold points of $\Mell$. For $\tau$ on the unit circle, the unit cell is a rhombus with a $\ZZ_2$ symmetry swapping $1$ and $\tau$. At $\tau = i$, it is a square lattice with a further $\ZZ_4$ symmetry, making a $
\ZZ_4$ orbifold point. Likewise  $\tau=e^{2 \pi i /3}$ defines an equilateral triangular lattice, which gives a $\ZZ_3$ orbifold point. Finally, $\Mell$ has a cusp for $\tau \to \infty$. In this limit, the elliptic curve becomes nodal. This picture is summarized in Figure \ref{fig:moduli space of elliptic curves}.

\begin{figure}[htp]
    \centering
    \includegraphics[width = \textwidth]{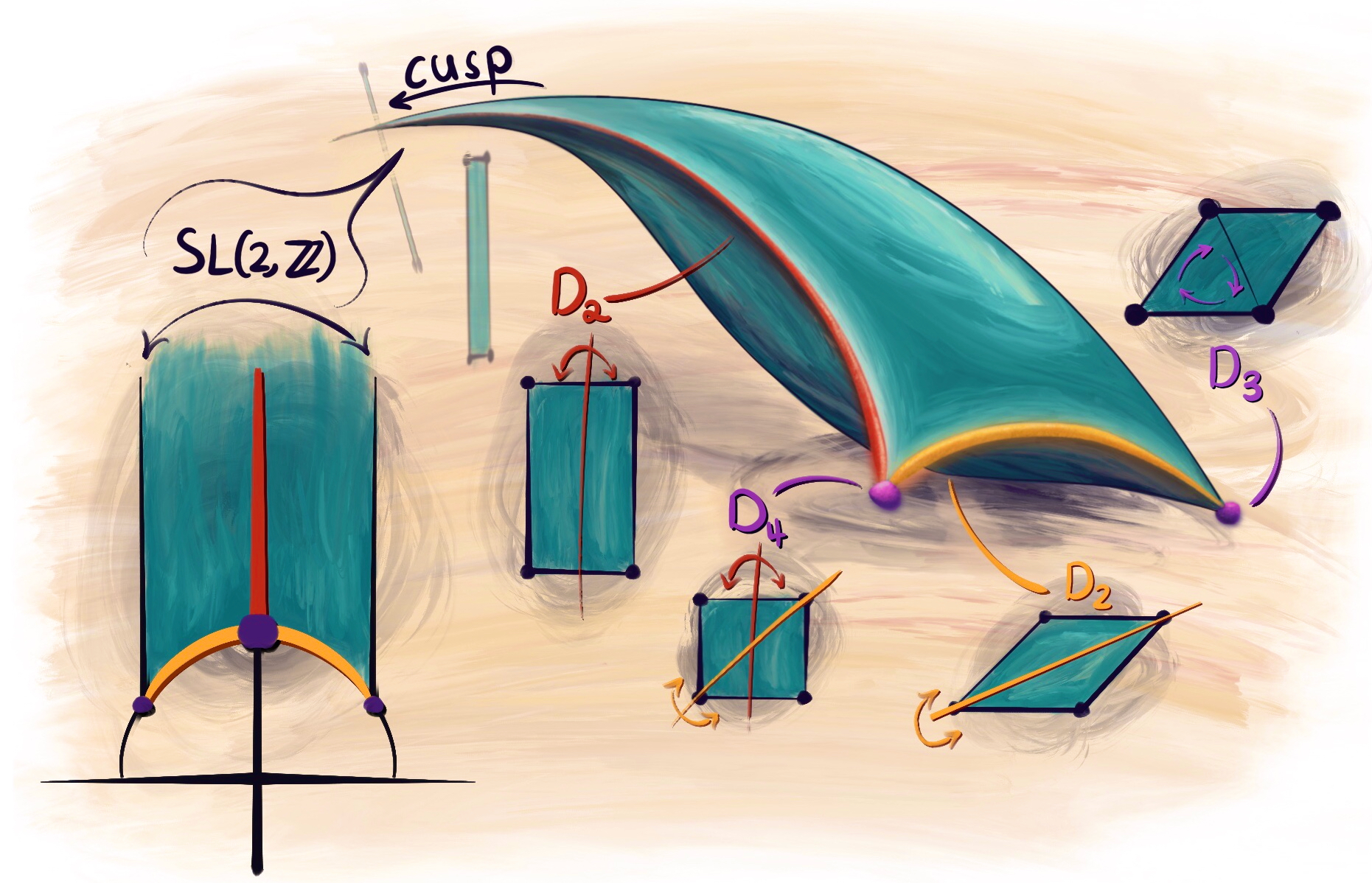}
    \caption{Illustration of the moduli space of elliptic curves $\Mell$, formed from a quotient of the upper half-plane by the modular group $SL(2,\ZZ)$. The orbifold points of $\Mell$ come from lattices with large isomorphism groups. These include rectangular lattices (red) and rhombuses (yellow), which have an extra symmetry of $\ZZ_2$ compared to the generic lattice. These are $\ZZ_2$ orbifold points or "folds" in $\Mell$. Additionally, the square lattice has extra $\ZZ_4$ symmetry, and the triangle lattice has extra $\ZZ_3$ symmetry, marked in $\Mell$ by purple dots. These are cone points with cone angle $\pi/2$ and $2\pi/3$, respectively. As the lattice parameter goes to infinity, the unit cell gets longer and longer, and $\Mell$ has a cusp.}
    \label{fig:moduli space of elliptic curves}
\end{figure}

\subsection{Description of band structure} \label{sec: genus 1/band}

The band structure of Euclidean crystals is well understood, permiting us to explicitly describe the band structure over the whole moduli space of Higgs bundles. For simplicity, let the potential equal zero (the empty lattice approximation). First, fix an elliptic curve with parameter $\tau$. If the crystal $\Sigma$ carries metric $g$, the crystal momentum space $\Jac(\Sigma) \cong \Sigma$ carries the dual metric $g'$. This is the standard metric on $\CC$ quotient the reciprocal lattice $\Gamma^* = \langle 1, \tau/ ||\tau||^2 \rangle = \langle 1, -1/\tau \rangle$. Note that $\tau \to -1/\tau$ is  an element of $SL(2,\ZZ)$, so the dual lattice is equal up to scale to the original. As Riemann surfaces, $\Sigma$ equals $\Jac(\Sigma)$, as elliptic curves are self-dual abelian varieties. For a free particle, the dispersion relation is quadratic, so  $E(k) = ||k||^2$. Including periodicity, these parabolas are centered on each lattice point in $\Gamma^*$ resulting in a many sheeted band structure over the unit cell. Level crossings occur at any point equidistant from two lattice points, with energy at crossing given by the square of that distance. The 2-torsion points, given by the half lattice $\Gamma^* / 2$, are always at least 2-fold degenerate. For example, the lattice center is always equidistant from kitty-corner lattice points.

The band structure over each fiber is identical, locally gluing together trivially on $\Mtoy$. Note that the fibrations carries the natural flat Gauss-Manin connection (in this case, the Picard-Fuchs connection), which preserves this band structure. The monodromy around the central fiber acts by elliptic involution, sending $L$ to $L\inv$. Luckily, the band structure is also invariant under this transform. There is a well-defined limit approaching the central $\PP^1$ of the nilpotent cone, which was obtained from $\Jac(\Sigma) / (L \to L\inv)$. 

There are a couple of c onsistant ways to assign a band structure to the satellite $\PP^1$s. The na\"\i ve choice assigns a constant band structure, equal to the spectrum above the associated 2-torsion point of the Jacobian. Alternatively, we can consider the band structure as a Bloch variety $\beta \times \CC$ fibered over $\Jac(\Sigma) \times \CC$. The base space becomes $\Mtoy$ after quotienting by the action of $\pm 1$, and blowing up the fixed points. We can define the band structure over the exceptional divisor to be the blowup of the Bloch variety $\beta \times \CC / \{\pm 1\}$ over those same points. This can resolve singularities in the Bloch variety. In particular, the band structure often has band crossings at time reversal invariant momenta, like the Dirac points in the band structure of graphene. These are generically conical singularities in the Bloch variety, so their blowup is smooth. We currently lack a good crystallographic interpretation of the singular fibers, so we cannot discern the proper choice of band structure.

Next, let us vary the elliptic curve parameter around $\Mell$. There is extra branching when the lattice is symmetric: In a rectangular lattice, the center is equidistant from all lattice points in the unit cell, meaning it is at least 4-fold degenerate. This occurs along the $\ZZ_2$ orbifold points of $\Mell$. In general, the band structure on the fibration of $\Mtoy$ with $\Mell$ branches along the orbifold points of $\Mell$, as these are higher symmetry lattices.

Approaching the cusp of $\Mell$ stretches the lattice to infinity in the imaginary direction. The dual lattice limits to the singular lattice $\langle 1,0\rangle$, and the band spacing approaches zero. At the limit, there is a continuous band along the the unbounded direction of the unit cell, and discrete bands along the bounded direction.

\subsection{Band structure with complex momenta}
Now we will look at the band structure in the complex momentum interpretation. This assigns a new Hamiltonian to the parabolic Higgs bundle on $\PP^1$. Namely, we add a factor of $i(\phi + \phi^\dagger)$ to the flat parabolic connection. In the parametrization of equation \ref{eq:explicit_higgs}, the connection is $\nabla^u + iB(\Phi^u + {\Phi^u}^\dagger)$ for $B\in \CC $ a point on the Hitchin base. Its Laplacian has the same spectrum as the Laplacian from a rank 1 Higgs field $(L,\phi)$ on the torus $\CC/\Gamma$. Since $\phi$ is holomorphic, it must be a constant one-form $B \de z$. The flat connection is $\de + A \de z + B \de \zbar $, where $A$ parametrizes the real part of the momentum. In this parametrization the Brillouin zone is $\CC/\Gamma* \times \CC$, where the first factor holds $A$ and the second $B$. We can write this a point as a vector of complex numbers $\vec{k} = (k_x,k_y)$, where $k_x,k_y$ hold the complex crystal momenta in each direction. $\Gamma$ only acts on the real part of $\vec{k}$. The free particle complex dispersion relation is $E(k_x,k_y) = k_x^2 + k_y^2$. Decomposing $\vec{k}$ into real and imaginary parts $\vec{k}_r, \vec{k}_i \in \RR^2$, we have
\[E(\vec{k}) = \left( ||\vec{k}_r||^2 - ||\vec{k}_i||^2  \right) + 2i \left( \vec{k}_r\cdot \vec{k}_i  \right).\]

The band structure is the graph of these dispersion relation in $\CC\times \CC$, centered around the points in the integer span of $k_x = 1+0i$ and $k_y = -1/\tau+0i$.  On the real locus, this equals the band structure described in Section \ref{sec: genus 1/band}. Off the real locus, as $B$ gets larger, the real energy gets smaller and smaller while the imaginary energy blows up, sometimes positive and sometimes negative. We note that the complex momentum Hamiltonian contains strictly more information than the crystal moduli one, as the crystal moduli had no dependence on $B$. This is not generally true of hyperbolic crystals.

\section{Speculations}
We have so far argued that Higgs bundles appear naturally in hyperbolic band theory. As Higgs bundles sit at the center of a dense web of mathematics and physics, it is conceivable that hyperbolic band theory may have manifestations in, or at least connections with, other important themes.  In this section, we will speculate on these connections and their possible implications.

\subsection{Connections to supersymmetric field theory }\label{sec:Speculation/High energy}

Higgs bundles already enjoy a deep involvement with physics, especially in regards to various high-energy considerations. This connection is most familiar from the work of Seiberg and Witten, where a Hitchin-type moduli space arises in the low energy limit of $\calN=2, d=4$ supersymmetric Yang-Mills (SYM) theory. This suggests an analogy between such theories and hyperbolic crystals. We recall the dictionary between supersymmetric field theories and Hitchin systems in the first two columns of Table \ref{table:dictionary}. (See \cite{neitzke_hitchin_2014} for an overview and \cite{donagi_seiberg-witten_1998} for a more detailed discussion.) The last two columns give the dictionary between Hitchin systems and hyperbolic crystals in the crystal moduli interpretation. By the transitive property of dictionaries, we relate supersymmetric field theories to hyperbolic crystals, mapping the Seiberg-Witten curve of an effective theory to a hyperbolic crystal.

\begin{table}[!h] 
\caption{Dictionary between $\calN=2$ SYM theory, Hitchin systems, and hyperbolic band theory.}
\label{table:dictionary}
\begin{tabular}{c|c|c}
\hline
\multicolumn{1}{|c|}{$\mathcal{N}=2$ SYM theory} & Hitchin systems               & \multicolumn{1}{c|}{Hyperbolic band theory} \\ \hline
Moduli of vacua                                  & Hitchin Base                  & Moduli of hyperbolic crystals                          \\
Seiberg-Witten curve                             & Spectral curve & Hyperbolic crystal                          \\
Lattice of EM charges             & Regular Hitchin fiber        & Abelian Brillouin zone                  
\end{tabular}

\end{table}

This analogy may be more explicit in the large Higgs field limit. This was studied using supersymmetric field theories by Gaiotto, Moore, and Neitzke, giving a conjectural picture of the asymptotics of the Hitchin metric \cite{gaiotto_wall-crossing_2013}. In the complex momentum interpretation of hyperbolic crystals, the large-scale limit is a large (imaginary) momentum limit, a semiclassical regime. This limit should be amenable to the geometric WKB method. For a $2$-sheeted cover represented by a quadratic differential, a preliminary analysis suggests the semiclassical dynamics restrict to trajectories of the quadratic differential, and the Hamiltonian has deep potential wells at the poles of the Higgs field. The large momentum limit becomes a tight-binding limit. It is represented by a quiver, with nodes for each pole and arrows connecting nodes according to the trajectories of the quadratic differential. This quiver is very similar to the BPS quiver of the field theory, which gives the BPS spectrum / Donaldson-Thomas invariants of an associated Calabi-Yau 3-fold.

The BPS spectrum arises from a quiver quantum mechanics problem, consisting of ground states of $\calN=4$ supersymmetric quantum mechanics on the moduli space of quiver representations. When the quiver embeds on the torus, there is a natural $U(1)^2$ action on the moduli of quiver representations, multiplying every cycle around the torus by a phase. The supersymmetric ground states localize to fixed points of the $U(1)^2$ action on this moduli space \cite{mozgovoy_noncommutative_2010, ooguri_crystal_2009}. Note that this $U(1)^2$ action is precisely the modification of the Hamiltonian under changing the crystal momentum of a Euclidean crystal. All together, the structures from supersymmetric field theories arise naturally in hyperbolic band theory, like spectral networks, their associated quiver, and the action on those quivers. It will be interesting to turn these suggestive associations into a firmer bridge between the subjects.

\begin{remark}
The fixed points of the torus action are uniquely described as a crystal melting state \cite{ooguri_crystal_2009}. At least \emph{a priori}, these crystals are unrelated to hyperbolic crystals. Similarly, the quiver quantum mechanics problem used to find the BPS spectrum is essentially unrelated to our quantum mechanics problem of finding the spectrum of a Hamiltonian.
\end{remark}

Regarding connections with quiver moduli, it is noteworthy that certain Higgs moduli spaces --- ones closely related the toy model in Section \ref{sec:genus 1} --- can be obtained in at least one complex structure as Nakajima-type quiver varieties, e.g. \cite{fisher_hyperpolygons_2016,rayan_moduli_2021}.  To be more precise, the quiver variety for the so-called ``comet-shaped'' quiver and an open part of the (parabolic) Higgs moduli space for arbitrary genus and certain flag types coincide as complex varieties.  Both moduli spaces are hyperk\"ahler but the complex-analytic isomorphism only holds in the $I$ complex structure.  The $J$ and $K$ complex structures on the quiver side are linearizations of the (multiplicative) character variety on the Higgs side.  There is also a connection between the Hitchin integrable system and a natural one on the quiver variety, which is roughly speaking of Gelfand-Tsetlin type (also as described in \cite{rayan_moduli_2021}).  The spectral data and integrable system on the quiver variety may also have a band-theoretic interpretation for both Euclidean and hyperbolic crystals, with the expectation that the Euclidean picture will correspond with the loop-free (``star-shaped'') variety.

\subsection{Topological materials}\label{sec:Speculation/Topological materials}
Topological materials have attracted a flurry of interest in recent decades (e.g. \cite{hasan_colloquium_2010}) and are, in part, an inspiration forthe more general theory of hyperbolic matter. These topological materials, which are gapped materials characterized by their vector bundle of eigenstates, enjoy properties that are protected by topological invariants. Some of these invariants are necessarily trivial in small dimensions, like the second Chern class, which is a $4$-form and thus vanishes on the $3$-dimensional momentum space of physical Euclidean crystals. Hyperbolic crystals have a momentum space with an arbitrarily large dimension, so the classes of topological insulators are unrestricted. This suggests an alternatives to constructing a model $4$-dimensional Euclidean crystal \cite{dwivedi_bulk_2016}, or using synthetic dimensions \cite{yuan_synthetic_2018}.  Two examples of hyperbolic topological insulators were constructed in \cite{boettcher_insulators_2022}.

These topological invariants exist over any space parametrizing Hamiltonians. Following the theme of this paper, we may consider these invariants over the moduli space of Higgs bundles. This space deformation retracts to the fiber above $0$ in the Hitchin base (the nilpotent cone), which thus encodes the topology of the entire space. So, any topological invariant is determined from the band structure over the nilpotent cone. In the crystal moduli interpretation, this suggests all topological properties are measurable in the degenerating limit of crystals.

These bulk topological insulating properties suggest a bulk-boundary correspondence with edge states at the boundary of a hyperbolic crystal.  These have been realized in a couple of fundamental hyperbolic models \cite{boettcher_insulators_2022}. The complex momentum interpretation suggests that edge states are characterized by Higgs bundles, as discussed in section \ref{sec:bulk-boundary}.

\subsection{High symmetry momenta} \label{sec:Speculation/Symmetry}
Physical crystals often have rich point group symmetries. Much of the standard band theory canon focuses on band crossings at momenta fixed under large subgroups of the point group. For hyperbolic crystals, the point group consists of automorphisms of the associated Riemann surface \cite{maciejko_hyperbolic_2021}. These act on vector bundles via pullback, defining automorphisms on the moduli space of vector bundles (crystal momentum space). For degree zero line bundles, the Torelli theorem says any polarization\footnote{The crystallographic meaning of the principle polarization of the Jacobian or its associated theta divisor is still unclear.} preserving automorphism of $\Jac(\Sigma)$ comes from an automorphism of $\Sigma$, possibly composed with the Kummer involution $L \mapsto L^*$.
On the Jacobian's universal cover $\CC^g$ the Kummer involution acts by the inversion $k \mapsto -k$. On the vector potential $A$ associated with $L$, the Kummer involution acts by $A \mapsto -A$, representing time reversal. Since the magnetic field is zero, the Hamiltonian is time-reversal invariant, and the Kummer involution preserves the band structure. Its fixed points, called the Time Reversal Invariant momenta (TRIMs) by physicists or 2-torsion points of the Jacobian by mathematicians, play a distinguished role in band theory. On the universal cover $\CC^g$, 2-torsion points are momenta $k$ with $2k$ a lattice point, so they are face centers of the Brillouin zone. For spin-$1/2$ systems, Kramer's theorem states the Hamiltonian is at least $2$-fold degenerate at such points. We can extend these symmetries to the moduli space of Higgs bundles $\MHiggs$. In the crystal momentum interpretation, each fiber is the Jacobian of a crystal, so fiberwise time-reversal (the Kummer involution) extends to all $\MHiggs$. For rank 2 trace-free Higgs bundles, this acts by holomorphic involution $(E,\phi) \to (E,-\phi)$ \cite{schaposnik_monodromy_2013}. Its fixed points are the $SL(2,\RR)$ Higgs bundles (first described in \cite[\S10]{hitchin_self-duality_1987}), making theses highly symmetric points of momentum space $\MHiggs$. One connected component is the Teichm\"{u}ller section, which corresponds to a trivial line bundle on each spectral curve, thus the zero momentum point of each Jacobian. These loci were discussed for Euclidean crystals in Section \ref{sec:genus 1/moduli space}

In general, every automorphism of the base curve gives an automorphism of the Hitchin moduli space, whose fixed points are highly symmetric momenta. These loci are part of an extensive line of mathematical research, see for instance \cite{schaposnik_monodromy_2013,garcia-prada_involutions_2018,carocca_jacobians_2006, rojas_group_2005,garcia-prada_involutions_2019}. They are generally branes on $\MHiggs$, submanifolds classified by their relation to the three complex/k\"{a}hler structures composing the hyperk\"{a}hler metric on $\MHiggs$ \cite{heller_branes_2018}. This symmetry is reflected in the band structures, connecting band theory with the hyperk\"{a}hler geometry of $\MHiggs$. In view of Section \ref{sec:Speculation/High energy}, the connection with the hyperk\"{a}hler metric could interface with the work of Gaiotto, Moore, and Neitzke.

\subsection{Geometric Langlands correspondence}\label{sec:Speculation/Geometric langlands}

The space of abelian crystal momenta for a genus $g$ surface is a $g$-dimensional complex torus $\Jac(\Sigma) = \CC^g/\Lambda$. This torus is also the momentum space for the dual $2g$-dimensional real torus $\RR^{2g}/ \Lambda^*$, which represents a Euclidean crystal in $\RR^{2g}$ with lattice $\Lambda^*$. This suggests an enticing opportunity to encode a high-dimensional Euclidean crystal inside a hyperbolic crystal. This duality has a new flavor in the crystal moduli interpretation, which represents the space of crystal momenta with a Lagrangian torus fibration. SYZ mirror symmetry gives a map between the moduli space of $G$-Higgs bundles and that of $^LG$-Higgs bundle where $^LG$ is the Langlands dual group, sending each fiber to its dual abelian variety \cite{hausel_mirror_2003,donagi_langlands_2012}. Our framework has $G=SL(n,\CC)$, but for general $G$ this map plays a central role in the geometric Langlands correspondence \cite{frenkel_gauge_2009}.  Physics has already found its way to geometric Langlands through S-duality of $\calN=4$ supersymmetric gauge theories \cite{kapustin_electric-magnetic_2006}. The structures on each side of the geometric Langlands correspondence are interpreted as branes on the moduli space of Higgs bundles, mapped to one another by mirror symmetry. Branes already arise in hyperbolic crystallography as high-symmetry crystal momenta (Section \ref{sec:Speculation/Symmetry}), perhaps bringing the geometric Langlands correspondence into hyperbolic band theory's sphere of influence.

\subsection{Fractional quantum Hall states} \label{sec:Speculation/FQHE}

With the basic framework of hyperbolic band theory established, the logical next step is to apply a magnetic field  \cite{stegmaier_universality_2021,ikeda_hyperbolic_2021,aoki_algebra_2021}. Our techniques should carry over without much trouble to integer quantum Hall states, where the magnetic flux through a unit cell is an integer multiple $N_B$ of the hyperbolic area. The magnetic field replaces degree zero bundles with degree $N_B$ bundles. (Alternatively, vector bundles of nonzero degree arise from projective unitary representations of the fundamental group \cite{atiyah_yang-mills_1983}, relating to the magnetic translation operators constructed in \cite{aoki_algebra_2021}.) Fractional quantum Hall (FQH) states (where flux per area is rational) are trickier, necessarily requiring multi-particle interacting states. One approach uses approximate ground states called Laughlin wavefunctions, which capture the qualitative aspects of FQH states, like quasiparticles with fractional statistics.
On compact Riemann surfaces, Laughlin states are holomorphic sections of a line bundle on the $n$\ts{th} symmetric product of the surface \cite{klevtsov_geometry_2016,klevtsov_quantum_2017}. This line bundle is twisted by the $n$\ts{th} exterior tensor product of an abelian crystal momentum \cite{klevtsov_geometric_2021}, fitting Laughlin wavefunctions into the framework of hyperbolic band theory. 
What about nonabelian crystal momenta? Perhaps these correspond to nonabelian FQH states on $\Sigma$, where the quasiparticles are vector-valued. These are best described with an effective Chern-Simons theory on $\Sigma$. Nonabelian Laughlin wavefunctions are the conformal blocks of a Chern-Simons theory \cite{moore_nonabelions_1991}. The space of conformal blocks is isomorphic to the holomorphic sections of the prequantum line bundle over the Chern-Simons phase space \cite{schottenloher_mathematical_2008}, which is diffeomorphic to the moduli space of trace-free Higgs bundles (for complex C-S theory). The holomorphic sections of the prequantum line bundle are nonabelian theta functions, which restrict to ordinary theta functions on each fiber \cite{beauville_spectral_1989}.

The relationship with Higgs bundles extends to effective field theory interpretations of FQH states. Quasiparticles of these theories are solutions to a vortex equation, which arise from Hitchin equations on subsets of the moduli space of Higgs bundles (for example, the abelian vortex equations are Hitchins equations on the Teichm\"{u}ller component \cite[\S11]{hitchin_self-duality_1987}). These effective field theories have a particle-vortex duality, replacing vortex solutions with fundamental fields and vice versa. This is a version of the S-duality of Chern-Simons theory, which manifests the geometric Langlands correspondence \cite{kapustin_electric-magnetic_2006}, and relates the integer and quantum Hall effects \cite{ikeda_quantum_2018,ikeda_hofstadters_2018}. Connecting with Section \ref{sec:Speculation/Geometric langlands}, perhaps two hyperbolic particle-vortex dual states are related by mirror symmetry of the moduli of Higgs bundles.

There are seemingly two ways to obtain FQH states on hyperbolic crystals. We can follow physics by considering interacting multi-particle states and studying their effective field theory. On the other hand, we can construct families of crystals from spectral covers then construct theta functions on each crystal's Jacobian, and then finally glue these together to a nonabelian theta function on Hitchin moduli space. Guided by aesthetic sensibilities, we guess that the resulting states are in some sense \textit{the same}.  (The crystallographic interpretation of theta functions is, itself, an intriguing open question.)

\subsection{Putting everything together}
The objects of our speculations already weave together into a well-connected web (Figure \ref{fig:flowchart}, black lines). We suggested how hyperbolic matter ties in with each (Figure \ref{fig:flowchart}, blue lines), and one can guess at the overall connectivity of this picture. Above, we speculated that both paths from the center to fractional quantum Hall states are somehow equivalent. It would be intriguing if this held more generally for every path through the web.  For instance, one may hope that the high-symmetry branes in $\MHiggs$ agree with the asymptotic hyperk\"{a}hler structure seen in the semiclassical large complex crystal momentum limit. At the same time, high-symmetry branes may map to one another as the Jacobian maps to the dual position space crystal, reflecting the geometric Langlands correspondence.  Further still, this map may relate two families of hyperbolic crystals, whose associated fractional quantum hall states are related by particle-vortex duality --- which is to say that perhaps Figure \ref{fig:flowchart} commutes.

\begin{figure}[htp]
    \centering
    \includegraphics[scale = 0.65]{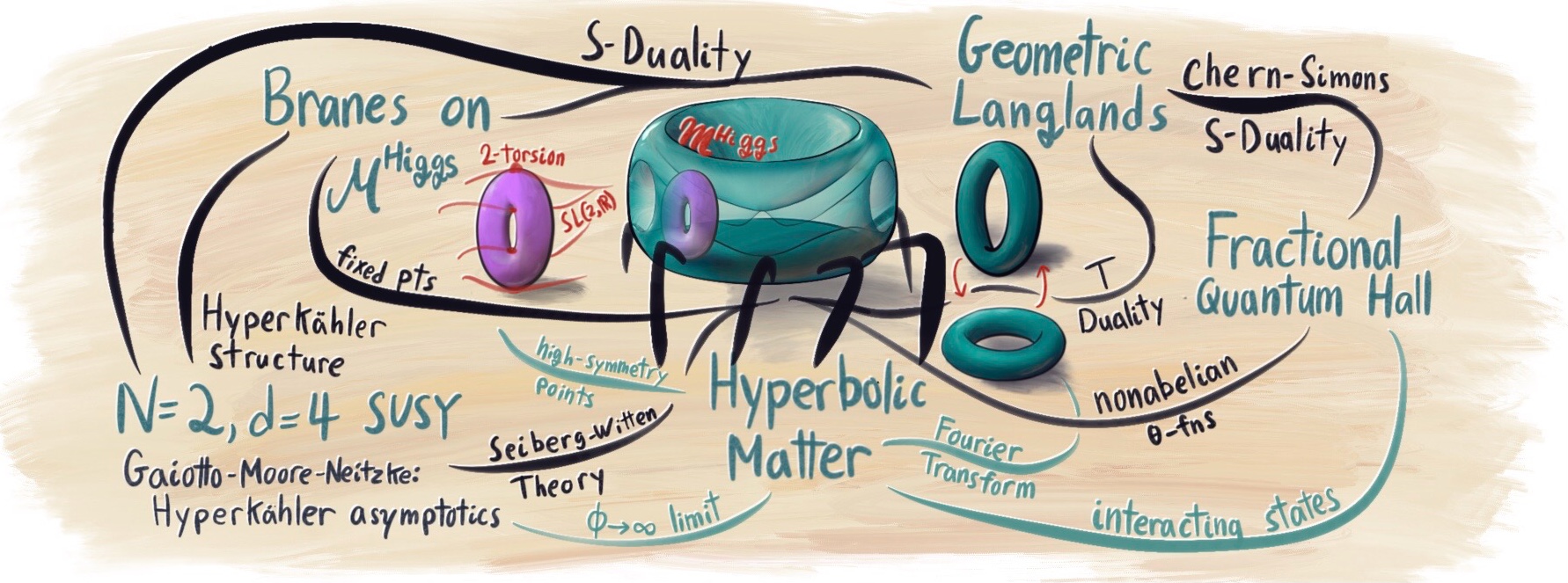}
    \caption{Flowchart summarizing the speculated connections between hyperbolic matter and various other objects (blue lines). It also shows preexisting connections between these objects (black lines). Most of these connections were woven by the character at the center of the web, namely Higgs bundles. }
    \label{fig:flowchart}
\end{figure}

The only evidence so far is aesthetic, and so we cannot call this a conjecture in good conscience. Instead, we conclude with:
\begin{pipeDream}
 This diagram commutes.
\end{pipeDream}
 
\pagebreak

\printbibliography
\end{document}